\documentclass{amsart}
\usepackage[utf8]{inputenc}
\usepackage[english]{babel}
\usepackage{booktabs}
\usepackage{amsmath,dsfont,amsfonts}

\date{May 2020}
\newtheorem{teo}{Theorem}[section]

\newtheorem{rem}[teo]{Remark}

\usepackage[T1]{fontenc}
\usepackage{xcolor}
\usepackage{amsmath}
\usepackage{graphicx}
\usepackage{esint}
\usepackage{babel}
\usepackage{comment}
\usepackage{enumitem}
\setlist[enumerate]{leftmargin=5.5mm}

\usepackage{hyperref}
\hypersetup
{   bookmarks=true,
	colorlinks=true,
	linkcolor=blue,
	anchorcolor = blue,
	urlcolor=blue,
	menucolor=blue,
	citecolor=black,
	pdfauthor = {Etienne Wasmer},
	pdftitle = {mortality containment vs. economics},
	pdfcreator = {LaTeX, TeXmaker, pdfLaTeX, Hyperref}
}
\usepackage{floatrow}
\begin{document}
\title[Mortality containment vs. Economics Opening]{Mortality containment vs. Economics Opening:
{\\ \small Optimal policies in a SEIARD model} }
\author[A. Aspri, E. Beretta, A. Gandolfi, E. Wasmer]{Andrea Aspri}
\address{ Johann Radon Institute for Computational and Applied Mathematics (RICAM)}
\email{andrea.aspri@ricam.oeaw.ac.at}

\author[]{Elena Beretta}
\address{Department of Mathematics, NYU-Abu Dhabi, and Dipartimento di Matematica, Politecnico di Milano}
\email{eb147@nyu.edu}

\author[]{Alberto Gandolfi}
\address{Department of Mathematics, NYU-Abu Dhabi}
\email{albertogandolfi@nyu.edu}

\author[]{Etienne Wasmer}
\address{Department of Economics, Social Science Div. NYU-Abu Dhabi}
\email{ew1555@nyu.edu}

\maketitle

\begin{abstract}

We adapt a SEIRD differential model with asymptomatic population and Covid deaths, which we call SEAIRD, to simulate the evolution of COVID-19, and add a  control function affecting both the diffusion of the virus and GDP, featuring all direct and indirect containment policies; to model
 feasibility, the control is assumed to be a  piece-wise linear function satisfying additional constraints. We describe the joint dynamics of infection and the economy and discuss the trade-off between production and fatalities. In particular, we carefully study the conditions for the existence of the optimal policy response and its uniqueness. Uniqueness crucially depends on the marginal rate of substitution between the statistical  value of a human life and GDP; we show  an example with a phase transition: above a certain threshold, there is a unique optimal containment policy; below the threshold, it is optimal to abstain from any containment; and at the threshold itself there are two optimal policies. We then explore and evaluate various profiles of various control policies dependent on a small number of parameters.

%In the simulations, an important point is that we take seriously the %constraints on policy: instead of assuming that lockdown can adjust %continuously, we drastically reduce the dimensionality of policy in %restricting it, more realistically, to a small number of degrees of %freedom and estimate the welfare cost of these restrictions (TBD).  

\end{abstract}
\newpage

\tableofcontents

\newpage

\section{Introduction}

A COVID-19 outbreak has begun in China at the end of 2019
\cite{HWO}, later spreading  to most other countries and causing 
a large number of infected individuals and deaths. In Italy, the first
country to be hit after China, the first confirmed
autochthonous case was recorded on February 21, and the first death on February 22 \cite{Stat2020};
the first death in US was recorded on February 28th \cite{NYT}.\footnote{These are the officially recorded dates, and the virus
might have been spreading before these times; we
record them here as references for the actual dates
we will use in simulations.}
The outbreak has so far caused at least $4$ million recorded
cases, and $275,000$ recorded deaths \cite{WM},
with real numbers estimated at much higher values.
In New York City there have been at this time at least 
$27,000$ deaths, corresponding to $0.335\%$ of the 
population. Massive regulatory responses have been put in place by most local and central governments, imposing restrictions (that we call lockdown hereafter) on travels and individual freedom. By several measures, the lockdowns have reduced the
spread of the virus and  the potential mortality.
On the other hand, the intensity of the impact of the pandemic, the lockdown policies, and  the behavioral response of agents beyond the regulations (spontaneous social distancing, etc.),
have greatly impacted the economic production. As of May 7, 2020, the IMF economic projection predict a loss in real GDP in 2020 of 3\% worldwide, as opposed to +3.45\% in the four years before (2016-2019). Even with the IMF forecast for the rebound of 5.8\% in 2021, the cumulated loss relative over the next two years \textit{relative to the trend} would be about 4\% of World GDP. In the advanced economies, this loss would be 5.65\%, including 6.85\% in the European Union and 5.8\% in the United States, and lower numbers in Asia and Pacific (-3.35\%) or Sub-Saharan Africa (-2.9\%). These are massive numbers, quite different by areas of the world, and updated regularly with likely higher GDP losses.

It is imperative for most regulatory bodies to balance
between the containment of the effects of the outbreak, and
the economic impact of the regulatory measures.
In this paper we adopt the number of COVID-19 fatalities
and the total GDP as proxies for the two effects,
and provide a framework to think about costs and benefits.
The two indicators
 have been selected for their reliability:
 GDP is a standard economic indicator,
 while mortality, in particular total mortality
 and its comparison with the expected mortality from
 previous years, is regularly monitored and 
 made public in many countries.
These assumptions allow to determine optimal lockdown policies
using optimal control theory.

More specifically, we consider a proxy for containment policies that encompasses the entire set of behavioral responses of agents who reduce consumption, the shut-down of markets themselves and measures that limit
people's movements, thus reducing the chances of
infection and the availability of labor. We then introduce the cost of a Covid related death for the social planner; for each intervention policy tuned
by a control function, we estimate a loss functional combining 
total Covid related fatalities and overall production loss in a given time frame.

The evolution of the epidemics is then described by a SEAIRD 
ordinary differential equations model, as specified in 
Section \ref{epidemic model} where a sizeable fraction of the population are asymptomatic individuals who can contaminate others. At each time $t$,
the lockdown is measured by an opening level of society (economic activity and social contacts) $c(t)
\in [0,1]$, $c=1$ being absence of any restriction and 
$c=0$ being the complete shutdown of all activities.\footnote{As a normalization, $c$ will be assumed to linearly affect  the infection rate and has a concave effect on GDP, see infra.}  

Many papers in the recent literature, including \cite{GKK}
 and various economic papers cited in Section \ref{LitRev} below,
compute the optimal policy in a general class with only technical 
restrictions on the policy space; but
this contrast with feasibility of the restriction policies,
which cannot adjust continuously: more realistically \cite{YZ},
restriction measures require a short time to 
be implemented, and then should be kept constant
for a certain time. For these reasons, 
we drastically reduce the dimensionality of 
the policy space, by taking controls
which are constant for some minimal period $\overline \delta$, and then transition linearly to the next level
in time $\underline \delta$. 
%We then
%compare to the optimal loss computed on much more
%general controls and estimate the 
%welfare cost of our restrictions.

The main point of our study is that one can 
find the various opening levels that avert a sizeable number
of deaths without determining an excessive damage to the economy:
Figure \ref{onelock} illustrates the potentials of
this analysis, in that deviations from the best
policies either cause an excessive economic loss for
a residual decrease in death rate, 
or an undesirably high mortality to prevent a rather
minor decrease in GDP. See for instance \cite{kaplan2020pandemics} for a similar assessment of the trade-offs involved, implicit or explicit in most economic works discussed in next Section. The darker blue curve in Figure \ref{onelock} reflects the constrained relation between mortality and GDP for different values of the control policy and can be thought as a technical rate of transformation. As we will explain later, it is generally preferable to be closer to the origin. A social welfare function and its indifference curves as in the light blue curve defines an optimal rule - when it exists. Its slope reflects the marginal rate of substitution (MRS) between mortality reduction and GDP losses and under simplifying assumptions, is the inverse of the statistical value of life, as we will explain later.\color{black}

 \vspace{0.4cm}
\begin{figure}[h!] \label{onelock}
    \centering
\includegraphics[scale=0.65]{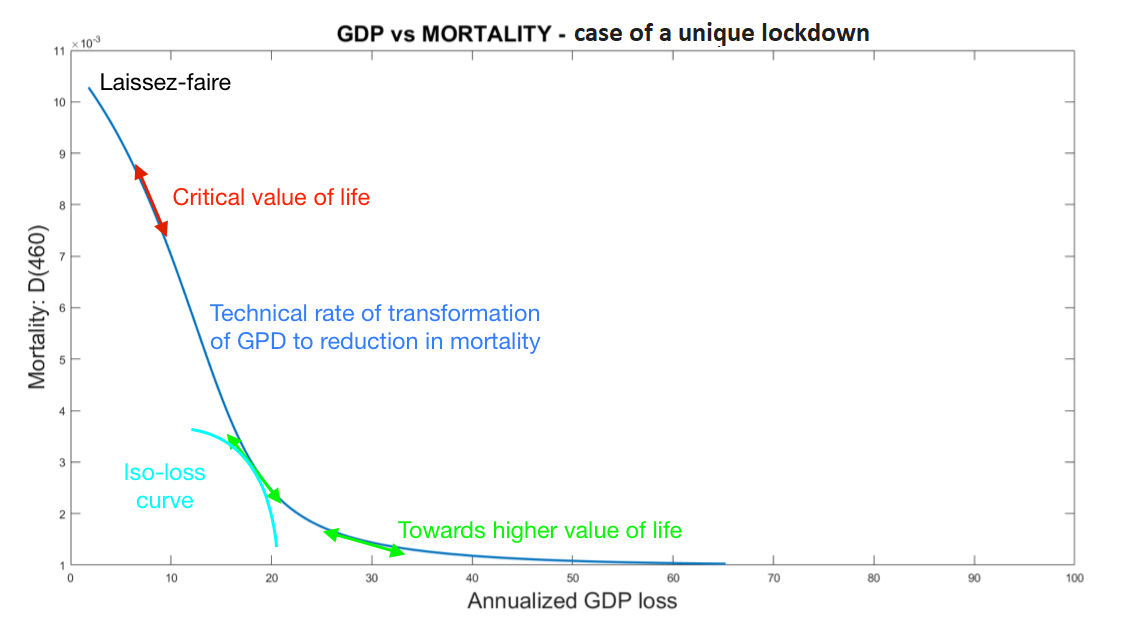}
    \caption{ Mortality  and production loss with one single, long lasting lockdown. The optimal choice,
    see Section \ref{SectUniqLock}, reduces mortality to $0.26\%$
    with a $19.45\%$ GDP loss: 
    the lockdown realizes a sharp containment of
mortality, but the constraint of protracted measures
causes a dramatic GDP loss.
This policy has not been followed
    by any country.}
    \label{DeathProductionLoss}
\end{figure}

  Figure \ref{onelock} also suggests that, for some values  of the preferred MRS, there could be two tangency points determining a  transition of phases, and possible non-uniqueness of solutions.
In Section \ref{uniqueness}, we show one simple example of this phenomenon, and
argue that it is the effect of a transition of phases in 
which the optimal control passes from being the
absence of any containment, to a more substantial 
tightening
as function of the social cost of COVID-19 deaths $a$.
As a result, at a critical value of $a$
there is coexistence of two optimal controls 
generating the same  value of the loss functional.
In addition, the multiplicity of suboptimal
controls around the critical value of $a$,
can have relevant social consequences
in terms on how to evaluate potential alternatives
to a given containment policy.
We then argue that, in general, the optimal control is likely to be unique provided that the 
social cost of COVID-19
mortality is large enough.
\color{black}

We finally consider some examples. Parameters are realistically
taken from current observations, and validated
by reproducing observed jump in mortality to this day,
and 
GDP reductions due to first lockdown periods. 
In the examples, we restrict, as mentioned, to
 simple controls in low dimensional spaces: 
 in the first example, a unique  lockdown is imposed
 at Day $85$
 till the end of the observed period, 
 a possibility that only few countries (such as Sweden, for instance)
 seem to have considered; with our choice of time frame,
 Day $85$ corresponds to March $25$: as detailed
 in Section \ref{calibration}, this is about when most lockdowns started; 
 in the second example, a partial reopening is realized 
 at Day $120$, after a
drastic initial lockdown has been imposed between Day 
$85$ and Day $120$, a typical situation at the current moment in time
in many countries;
in the third example, 
 a periodic alternation of lockdowns and reopenings
 is applied. In the last example, the optimal control 
 leads to herd immunity, which  is achieved 
 in such a way as to have
  very few infected at the time in which
 the immunity is reached;  the optimization has
 automatically determined the best possible
 access to herd immunity
  \cite{BM}. All examples are explicitly simulated and 
 optimal controls are numerically determined. 
 We then carry out a sensitivity analysis
 to evaluate the sensitivity of the results to 
 errors and fluctuations in parameter selection.
 %Finally, exploiting Pontriagyn principle, we compute the overall optimal strategy
 %in the larger class of Lipschitz functions,
 %and compare the welfare loss that is incurred
 %by restricting to  feasible, piece-wise linear
 %controls.
 
\color{black}
\section{Previous works and limitations}
\subsection{Brief literature review} \label{LitRev}

The number of papers adapting the SIR model to various economic contexts is large and rapidly evolving and it is impossible to make justice to the literature.

\cite{jones2020optimal} derive an optimal strategy where the social planner can affect both the contacts from consumption and contacts from production, each of them contributing to a third on the diffusion parameter $\beta$. They study the optimal policy using a standard growth model with leisure-consumption trade-offs. Agents react too little to the epidemics because they do not contemplate the impact of their behavior on other agents' infection rate and a lockdown seriously reduce infection and fatalities in flattening the curve, and avoid congestion of ICU units that would increase the fatality rate.

\cite{eichenbaum2020macroeconomics} study a standard DSGE model with a SIR contagion. They find that the epidemic causes \textit{per se} a moderate recession, with aggregate consumption falling down by 0.7\% within the year. Optimal containment would lead to a more drastic loss in consumption by 22\%. They also discuss the model with various health policies including vaccines, preparedness and other dimensions.

\cite{acemoglu2020multi} develop a multi-SIR model with infection, hospitalization and fatality rate depend on age, with three classes of individuals (young, middle-aged and old). They find that targeted containment policies are most efficient. For the same loss in GDP (-24\%), the targeted policies reduce mortality by 0.7 to 1.8 percentage points. They also include a stochastic vaccine arrival, not known for sure by the policy maker, and the stochastic process evolves over time. They assume as in \cite{alvarez2020simple} that full lockdown is not feasible, as we also assume. In \cite{alvarez2020simple} have a SIR model embedded in the growth model. Their optimal policy is to implement a severe lockdown 2 weeks after 1\% of the population is infected, to cover 60\% of the population, and then gradually reduce the intensity of the lockdown to 20\% of the population after 3 months. The absence of testing reduces the welfare. With testing and under the optimal policy, the welfare loss is equivalent to 2\% of GDP. Another paper on sequential lockdown with heterogeneous population is \cite{rampini2020sequential}. In particular, he uses a fatality rate of 0.06-0.08\% for younger agents and 2.67 to 3.65\% for older workers.

\cite{hall2020trading} study a variant with the minimization of an objective function and Hamilton-Jacobi-Bellman. Basing their fatality rate on 0.8\%from the Imperial college study, they argue that the optimal decline in consumption is approximately 1/3rd for one year. They then consider more recent estimates of the fatality rate, around 0.3\% across age groups, and argue that the optimal decline in consumption is still around 18\%. Our numbers are in line with these numbers.      

\cite{gollier_slides}, similar to us
assumes that a vaccine is ready after a few months (52 weeks in his case). He uses a $R_0$ around 2 (1.85 on the slides available on line) prior to containment, and the containment policy drives it down to 1, as we do. He uses a value of statistical life of 1 million euros and studies confinement scenarii under notably periodic reinfection rates. In \cite{gollier_paper_2, gollier2020cost}, he further explore the ethics of herd immunity and elaborate on lockdowns differentiated by age groups. In particular, he uses (Table 4 of \cite{gollier_paper_2}) a valuation of statistical lives depending on age, with the population between 60 and 69 representing 37\% of that of individuals below 19, the population between 70 and 79 representing 23\% and those above 80\% being slightly less than 10\% of that maximum value. He further discusses the critical moral hazard issues associated with the epidemic.

Economic consequences associated with demand and transmission mechanisms have been studied in \cite{guerrieri2020macroeconomic}: they show that in the presence of multi-sector production, with or without imperfect insurance, it is possible and plausible to have demand shocks in the second round going beyond the initial supply (shutdown shock). They study various aspects such as. labor hoarding and bankruptcy cascades. \cite{gregory2020pandemic} study the response of the economy in a search framework. The existence of search frictions slows down the recovery, and under reasonable parameter values, the initial lockdown strategy is likely to have long-lasting effects. In their baseline scenario, unemployment increases by 12 percentage points of the labor force for a year, and it takes 4 years to get back to 3 percentage points above the starting point before the lockdown. They find, interestingly, that it is better to have a longer initial lockdown (6 months) and no uncertainty that a shorter lockdown with the risk during 9 to 12 months to face a second lockdown. \cite{farboodi2020internal} estimate a SIR model in which the decline in activity comes from the optimal response of agents without intervention, and where immediate distancing in a discontinuous way, until a treatment is found, is a superior policy, to contain the reproduction number. In contrast, \cite{krueger2020taojun} calibrate a model similar to \cite{eichenbaum2020macroeconomics} in introducing goods that can be consumed at home rather than in public places and show that a Swedish-type policy of no-lockdown but strong behavioral response by agents reduces the socio-economic costs of Covid by up to 80\%.  

Last and most related to us, \cite{garibaldi2020modelling} analyze the existence of a SIR-matching decentralized equilibrium and analyze the inefficiencies stemming from matching externalities to determine the optimal way to reach herd immunity.

To conclude, in most of the papers cited above, there is an explicit focus on the optimal policy and the difference between the laissez-faire and the optimal policy is important, due to the externality of contagion. What our paper adds is a formal treatment of existence and a discussion of the potential multiplicity of solution and phase transition due to the non-linearity in the transmission mechanisms of the epidemic. Another paper in this spirit by \cite{lukaszrachel} finds explicit optimal solutions in a set of constrained policy functions and characterizes in particular the optimal starting date of the lockdown and discusses time-consistency issues.

\subsection{Limitations}

Our results are only a first indication of a modeling methodology for the search of an optimal trade off between containment of fatalities and reduced loss in welfare. While the parameters of the SEAIR model are related to the current outbreak, a more detailed model needs to consider stratified and geographically dispersed populations, and more elaborate lockdown policies, targeted to regions, industries and population that are more at risk. The following points are in order. 
\begin{enumerate}

\item As discussed above,  several papers have recently addressed similar questions, with in particular a focus on the optimal lockdown policy in the presence of behavioral response of agents on production, on investment or in consumption, of heterogeneity of the population and  on learning on the underlying parameters of the economy. 
Here, as usual in most current literature on COVID-19, 
we use an extension of the SIR model, hence assuming 
that each individual has the same chance of meeting
every other individual in the population.\footnote{This is 
a very limiting assumption, and can be well approximated only by small communities.  However, this assumption can also be seen as the equivalent of macroeconomic model with a representative agent. The parameters reflecting the aggregate behavior are not necessarily the parameters of the underlying individual agents, but are adjusted to fit the aggregate data in the best. This is a very similar discussion to that in \cite{keane2012micro} regarding labor supply elasticities.\color{black}} More realistically,
one would need to consider
geographically dispersed populations with long range 
interactions and communities (in the spirit of \cite{gandolfi2016} for instance).

\item In this paper, in order to have an accurate model of the dynamics of the pandemic with  several classes (susceptible, exposed, asymptomatic, symptomatic, recovered, fatality, natural demographic turnover), and yet be able to prove  existence and discuss conditions for uniqueness of an optimal response function, we treat the simpler case where the social planner can directly control the contagion parameter with an instrument that also affects GDP, either influencing the behavior of agents or closing markets.

\item The simulations we provide are based on parameters known at the time of this study, which are also the parameters perceived by policy makers at the time of decision making. With these parameters, we find that the statistical value of a human life 
that lead to the application of observed levels of 
lock down is in line with the value employed in 
actuarial sciences.

\item Given the nature of the virus and its novelty, there is some uncertainty surrounding the parameters, and these are likely to evolve as medical and epidemiological  research progresses. The final numbers will only be available gradually, with large testings currently being implemented. Our approach will therefore only 
allow us to reassess current policies retrospectively, in one way or another,  when the uncertainty at the time of decisions will have dissipated. 

\item Similarly, the parameters connecting the  spread of the diffusion of the virus to the loss of GDP from lockdown are uncertain. We choose a median way in the numbers in our simulations. 

\item We remain agnostic in our conclusions and provide sensitivity analysis in describing a range of alternative parameters. The shape of the optimal response in time is relatively invariant to those parameters, but warn that the intensity of the optimal lockdown relies a lot on exact numbers chosen in our simulations. 

\item On the economic side, one dimension not analyzed yet is the fact that the loss of GDP - a supply shock here - is likely to produce second round demand effects, leading to a persistence in the recession that our model does not take into account. Another limitation, of a similar spirit, is the ability of the lockdown to be reversible in the short-run, that is, once stopped, assembly lines may need a lag to resume.

\item Another limitation in the benchmark exercise is that the fatality rates vary enormously by age and morbidity, and in particular, the fatality rate is 10 times higher at least between the population below 60 and above 60. Since the lockdown mostly acts through adjustment of the labor force in our model, more analysis is needed to draw consequences about the overall lockdown strategy. We cannot deliver conclusions about the opportunity of the observed lockdown.  

\item Another limitation is that our model does not focus on the behavioral response of agents who may have learned about the parameters of the diffusion of the epidemics and reduced the 
infectivity of the virus independently of the lockdown. We do however believe that there are behavioral responses, but as in \cite{jones2020optimal}, we also believe that there are strong externalities in the contagion process that the purely-selfish individual behavior would not internalize. In that sense, the non-behavioral approach we follow is a proxy for the inefficiency of the decentralized equilibrium approach that leads to excessive contamination of the population. Future work should however relax the lack of behavioral response and investigate the size and sign of the interaction between government regulations and individual responses.

\item Last but not least, contrary to other studies, we limit our welfare analysis to a fixed period if time of the pandemic, one year and one quarter in the simulations. The implicit assumption is that after one year, treatments will have improved and vaccines may be possible. This acts as an extreme capitalization effect: in the future, technology will have improved and this is already integrated in economic calculation of the present time. It is easy to do  a sensitivity analysis where the length of time periods  is augmented, and investigate whether a new cycle of pandemic and lockdown is needed. The solution we exhibit for the optimal lockdown are therefore useful not only to rationalize the current experience, but also to prepare to the next wave or the next virus. We however introduce this assumption of a fixed and short period of time over which the smoothing occurs because the hope of a vaccine was present in public discussion.\footnote{As an example, the BBC reported on May 19, 2020 that the US company Moderna had been successful in training the immune system in human. The announcement lead to a 30\% increase int he value of this company in the stock markets. See
https://www.bbc.com/news/health-52677203}
\end{enumerate}

In Appendix, we present  optimal control problems that would address some of these limitations.

\section{A simple SEAIRD model with containment}
\subsection{Epidemic model } \label{epidemic model}

We consider SEAIRD, a  version of the SIR model \\(\cite{CHBC-C} (25) Page 20),
with some realistic features taken from current observations
of the Covid-19 outbreak.
The population is divided into: susceptible (S), exposed (E),
asymptomatic (A),
infected (I), recovered (R), Covid related deceased (D),
and natural deaths (ND).
Variables are normalized so that $S+E+A+I+R+D=1$. 
Overall,  we consider  a natural death rate $n$. This is compensated
by a natural birth rate, that  can be considered as
the rate of inclusion into the labor force; 
the natural birth rate is reduced by a factor that
can be interpreted as a Covid related slowdown.

We assume that affected individuals become first exposed (E),
a phase  in which they have contracted the virus and are contagious, without showing symptoms. Exposed individuals either develop 
symptoms at a constant
rate $\epsilon \kappa$, becoming infected, or progress into being asymptomatic 
till healing with rate $(1-\epsilon ) \kappa$.
A susceptible 
individual is assumed to have a uniform probability of encountering
every exposed and asymptomatic, and has a  probability of 
coming in contact with an infected severely 
reduced by a factor $s<1$. The parameter $s$ can be thought of as measuring the effect of an isolation policy that has \textit{per se} no direct effect on the labor force able to participate in economic production. Instead, the probability of all
encounters  is then affected by the mitigation policies via a factor $c(t)$, that will affect economic activity, as discussed in the next section.
Upon encounter, there is a
rate $\beta$ of transmission.

Those who are infected  recover at  rate $\gamma$, 
or do not recover and die at rate $\delta$; 
$\delta/\gamma$ is the   deaths to recovered ratio 
to be estimated from 
current available observations.
Asymptomatic recover at rate $\gamma$.

\begin{align} 
\text{Susceptible:\quad}\frac{dS}{dt} & =-\beta S c(t)(sI+E+A)-nS +n(1-D)\label{eq:dS}\\
\text{Exposed:\quad}\frac{dE}{dt} & =\beta c(t) S(sI+E+A)-(\kappa+n)E \label{eq:dE}\\
\text{Asymptomatic:\quad}\frac{dA}{dt} & = (1-\epsilon)\kappa E
-(\gamma+ n) A \label{eq:dA}\\
\text{Infected:\quad}\frac{dI}{dt} & = \epsilon\kappa E-(\gamma+
\delta+n)I\label{eq:dI}\\
\text{Recovered:\quad}\frac{dR}{dt} & =\gamma (A+I){-n R}\label{eq:dR}\\
\text{Covid deceased:\quad}\frac{dD}{dt} & =\delta I\label{eq:dD}\\
\text{Natural deaths:\quad}\frac{dD_N}{dt} & = n(S+E+A+I+R)\label{eq:dDN}
\end{align}
The initial population at the onset of the outbreak
of a previously unknown virus
consists primarily of susceptible, $S(0) \approx 1$, and a small fraction of exposed, so that $S(0)+E(0)=1$. For the model under consideration the  reproduction number  has the following expression 
\begin{eqnarray} \label{R0}
    \mathcal{R}(t)&=&\beta S(t) c(t)\left(\frac{1}{\kappa+n}+\frac{\kappa}{\kappa+n}\frac{(1-\epsilon)}{\gamma+n}+\frac{\kappa}{\kappa+n}\frac{s \epsilon}{\gamma+\delta+n}\right)\nonumber  \\
    &=&c(t) S(t)\frac{\beta  \kappa}{\kappa+n}\left(\frac{1}{\kappa}+\frac{(1-\epsilon)}{\gamma+n}+\frac{s \epsilon}{\gamma+\delta+n}\right) 
\end{eqnarray}
with basic reproduction number
$\mathcal{R}_0 = \mathcal{R}(0)$.
Notice that the population $S+E+A+I+R+D$ is preserved. This is a consequence of
the fact that by including the term $n$ demography replaces
all deaths except Covid deaths.\footnote{Mathematically, this is
easily seen 
by taking the derivative of $S+E+A+I+R+D$. In fact, letting
$\phi=S+E+A+I+R+D$, we have that $\phi(0)=1$  and 
$\frac{d\phi}{dt}=\frac{d(\phi-1)}{dt}
=-n(\phi-1)$, so that, since $(\phi-1)(0)=0$, necessarily $\phi \equiv 1$ by uniqueness of solutions of
differential equations.}

\subsection{Containment policies }
Containment policies are aimed at reducing the spread of the 
epidemic by reducing the chances of contacts among individuals. This is reflected in the model by a coefficient $c(t)$ that modulates
the encounters between susceptible and either exposed,
infected or asymptomatic
individuals. We assume that the reduction
is the same for all groups, as we have already included the
effect of symptoms in segregating infected individuals.
This justifies the factor $c(t)$ in \eqref{eq:dS}.

The opening level function $c(t) $ takes values in $[c_0,1]$
$c_0>0$;
$c(t)=1$ indicates that there is full 
opening, and no lockdown measures have been taken,
this is, by default, the status at the early stages of the outbreak. 
The lower bound $c_0$ corresponds to the infeasibility of
a complete shutdown;  this features the fact that there will always be a minimum amount of productive activity (e.g. 
via internet for home production) from private agents that cannot be interrupted.
Provided $c_0$ is small enough, all
our results are insensitive to the precise value. Further, to model concrete feasibility of the policy, the control is assumed to be a continuous, piece-wise linear function, with the
additional constraints of being constant for 
long enough time intervals $\overline \delta$; the transitions between the 
various  constant levels are taken to be linear
and last at least some $\underline \delta$ to model
non-negligible
friction in policies implementation; the controls are
then Lipschitz\footnote{A function $f$ is Lipschitz continuous with Lipschitz constant $M$ on an interval $[a,b]$ if there exists a constant $M$ such that $\frac{|f(x)-f(y)|}{|x-y|}\leq M$ for any $x,y\in [a,b], x\neq y$.} continuous.
The detailed form of $c(t)$ is given in Section \ref{4.1};
and several examples are presented in the Section \ref{Examples}. 

The class of containment policies considered in this work is 
in sharp contrast with other choices, such as \cite{GKK}, in which
all continuous functions are considered as possible controls.
Our work is in the spirit of other applied papers
\cite{RA}, focused on more realizable controls.

\section{Economic effects of epidemic and lockdown}

\subsection{Social planner's objective}

We investigate optimal containment policies balancing the 
effect of overall death vs. loss of production. 
This includes an a-priori evaluation of the social cost
of Covid deaths, embodied in a constant $a$.
The social planner's loss functional (the negative of its utility) $\mathcal{W}$ combines production $P$
and the number of new deaths from Covid \footnote{An interesting question is whether the social planner should also consider the 
change  in natural deaths due to
a decreasing population, a reduction of
traffic accidents and an increased 
risk for untreated pathologies caused by the
lockdown and the outbreak itself. We do not address this important question here.}  $D'(t)$, as follows: $\mathcal{W}=-\frac{P^{1-\sigma}}{1-\sigma}-aD'(t)$.
The social planner minimizes a loss function between an initial period
$t_0=0$ and a final period $t_{1}=T$ which could be infinity:

\[
\mathcal{L}=\left\{ \intop_0^Te^{-rt}\left[\mathcal{V}(P(t))+aD'(t)\right]dt\right\} 
\]
where  $\mathcal{V}(P(t))$
is a decreasing convex function of the GDP $P(t)$, and $a$ is the cost
of a covid death $D(t)$ for the social planner. 
The social planner discounts the future at rate $r$; 
such discount factor incorporates both the lesser interest for more distant economic consequences and the preference for containing immediate deaths,
hence it acts in the direction of
flattening the infection curve. Further normalizing the full-capacity GDP to 1, and 
assuming that the loss function is zero 
for full capacity, a typical function
would be:
\[
\mathcal{V}(P)=-\frac{P^{1-\sigma}-1}{1-\sigma}
\]
with $\sigma>0, \sigma\neq 1$, and $\mathcal{V}(P)=-\log(P)$
if $\sigma=1$. For values of $\sigma$ above 1 (our choice hereafter will be 2),  
\[
\lim_{P\rightarrow0}\mathcal{V}(P)=-\infty;
\]
 it follows that  $c=0$ is  never reached,
 and this further justifies the assumption of 
 $c\geq c_0$.
We have 
\begin{align*}
\mathcal{V}'(P) & =-P^{-\sigma} \\
\mathcal{V}"(P) & =\sigma P^{-\sigma-1}>0
\end{align*}

As a last remark, with a linear loss function $\sigma=0$, the parameter $a$ can directly be interpreted as the value of life in elasticity with respect to GDP. With higher values of $\sigma$, the value of $a$ relates to the value of life in marginal utility of GDP, given the aversion to intertemporal fluctuations in GDP that is characterized by  the elasticity of intertemporal substitution introduced in the next section.

\subsection{Production and welfare }

We take the overall production $P$ to be a linear function
of labor. At any given time, the  labor force  is $S+E+A+R$,
but its effective availability for production is
determined by the current opening policy  $c(t)$.
The link between  $c(t)$ and GDP is captured by
a function
\[
\mathcal{G}(c(t))
\]
and it affects GDP as:
\begin{align} \label{production}
P(t) & =\mathcal{G}(c(t))L(t)\\
 & =\mathcal{G}(c(t))\left[S+E+A+R\right].
\end{align}

 Labor availability
in the presence of a lock down is not assumed to be linear, as
the effects of socio-economic restrictions can be contained by 
work force substitution or increased productivity. We assume an iso-elastic control
%\begin{align}
  %  \mathcal{G}(c(t))=\frac{(c(t))^{1/3}+(c(t-\tau))^{1/3}}{2}
%\end{align}
\begin{align}
    \mathcal{G}(c(t))=c(t)^{\theta}
\end{align}
with $\theta\in (0,1)$ for reasons discussed in the parameter selection  section \ref{calibration}. We think of $\theta$ as a reduced form parameter that connects the infection spread and the change in GDP.

With these assumptions, the loss function becomes

%\begin{align}
%\mathcal{L}=\left\{ \intop_{t_{0}}^{t_{1}}e^{-rt}\left[
%-\frac{(\bar{\beta} \frac{(c(t))^{1/3}+(c(t-\tau))^{1/3}}{2} %\left[1/2+S+E+R\right])^{1-\sigma}-1}{1-\sigma}
%+aD'(t)\right]dt\right\} 
%\end{align}

\begin{align} \label{loss}
\mathcal{L}=\left\{ \intop_0^{T}e^{-rt}\left[
-\frac{(c(t)^{\theta} \left[S+E+A+R\right])^{1-\sigma}-1}{1-\sigma}
+aD'(t)\right]dt\right\}. 
\end{align}

\section{Mathematical results}
\subsection{Existence of a global minimum of the loss functional}
\label{4.1}
In this section we  prove the existence of a global minimum over a suitable class of control functions $c$.
More precisely, fix two values $\overline \delta,
\underline \delta $ with $\overline \delta > 2
\underline \delta>0 $, and let 
$\mathcal{K}$ be the collection of continuous
functions 
$$c:[0,T]\rightarrow [c_0,1],
$$ 
such that there exist
$ \alpha_1< \dots < \alpha_{k-1}\in [0,T]$ and 
$c_0\leq \beta_1, \dots, \beta_k\leq 1$, with
$\alpha_{i+1}-\alpha_i \geq \overline \delta$
for all $i=1,\dots, k-1,$
such that $ c(t)$  is continuous and 
\begin{eqnarray}\label{piecelin}
c(t)=\begin{cases}
\beta_1 \quad \text{ if } t \in [0,\alpha_1]\\
\beta_i \quad \text{ if } t \in [\alpha_{i-1}+\underline \delta,\alpha_i], \quad
i=1, \dots, k\\
 \beta_i+(\beta_{i+1}-\beta_i)(t-\alpha_i)/\underline \delta\quad \text{ if } t \in [\alpha_{i},\alpha_i+\underline \delta],\\

\end{cases}
\end{eqnarray}
where we have taken $\alpha_0=0, \alpha_k=T$.
\begin{figure}
    \centering
    \includegraphics[scale=0.3]{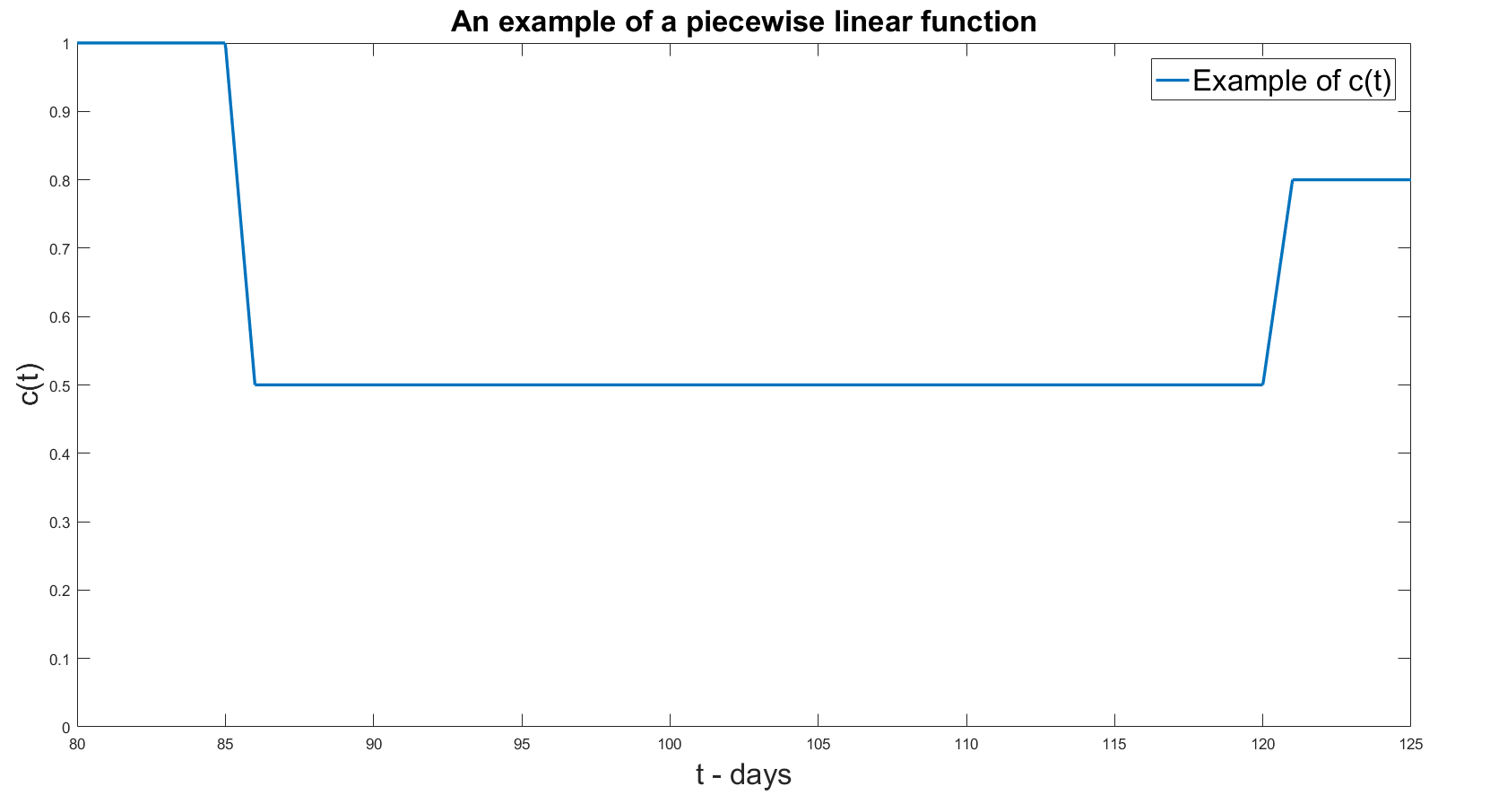}
    \caption{An example of the control variables $c$ used as allowed opening levels. Notice that we always consider controls of this type,
although in the figures presented in 
later sections the linear, non costant, portions might not be
easily detectable. }
    \label{fig:piecewise_linear}
\end{figure}
Notice that $\mathcal{K}$ is a class of Lipschitz continuous functions with Lipschitz constant bounded uniformly by $(1-c_0)/\underline \delta$ on $[0,T]$, as examplified in Figure \ref{fig:piecewise_linear}.
%that we will denote with $L_M([0,T])$
\begin{teo}
$\mathcal{K}$ is relatively compact in the space of continuous functions $C[0,T]$.
\end{teo}
\begin{proof}
For each sequence $\{c_n\}, c_n \in \mathcal{K}$,  we have 
$c_n\leq 1$ and $|c_n(x)-c_n(y)|\leq (1-c_0)|x-y|$; by Ascoli-Arzela Theorem, the sequence converges uniformly in $[0,T]$, possibly up the a subsequence,   to a continuous function $c$. Clearly the function $c$ has range in $[c_0,1]$ and is Lipschitz continuous with Lipschitz constant bounded by $(1-c_0)/\underline \delta$. Let us prove that it must be piece-wise linear of the form (\ref{piecelin}).

Consider $\eta \leq \underline \delta$ and points of the form 
$x_k=k \eta$ for $k=1,\dots,\lfloor T/\eta\rfloor$.
Take $k_1$ and $k_2$ such that $|k_1-k_2|<\underline \delta/\eta$.
 We consider the two possible cases.
\begin{enumerate}
    \item Suppose  $c(x_{k_1})=c(x_{k_2})$. Since
    $|x_{k_1}-x_{k_2}| <\underline \delta$, any $c_n$ has the extreme values of the interval $[x_{k_1},x_{k_2}]$ exactly 
    at $x_{k_1}$ and $x_{k_2}$;  
     for small
    $\epsilon$ and large enough $n$, assuming, without loss of
    generality, that $c_n(x_{k_1}) \leq c_n(x_{k_2})$, we have 
    $$
    c(x_{k_1})-\epsilon \leq
    c_n(x_{k_1}) \leq c_n(x) \leq c_n(x_{k_2})
    \leq c(x_{k_2})+\epsilon=c(x_{k_1})+\epsilon
    $$
    for all $x\in [x_{k_1}, x_{k_2}]$. Hence, 
    in the limit for $n \to \infty$, we have $c(x)=c(x_{k_1})$
    for all $x\in [x_{k_1}, x_{k_2}]$.
    \item If $c(x_{k_1})\neq c(x_{k_2})$,
    then take $x_{k_2}+\underline \delta$ and $x_{k_2}+\underline \delta+(\overline \delta-\underline \delta)/2 $:
    it must be $c(x_{k_2}+\underline \delta)=c(x_{k_2}
    + \underline \delta+(\overline \delta-\underline \delta)/2)$;
    in fact, for $\epsilon < |c(x_{k_1})- c(x_{k_2})|/3$
    and $n$ large enough, $|c_n(x_{k_1})- c_n(x_{k_2})|
    >|c(x_{k_1})- c(x_{k_2})|/3 >0$, hence $c_n$ must have
    a non constant part in $[x_{k_1},x_{k_2}]$ 
    and must thus be constant in $[x_{k_2}+\underline \delta,x_{k_2}+\underline \delta+(\overline \delta-\underline \delta)/2]$.
    For the same reason,  $c(x_{k_1}-\underline \delta)=c(x_{k_1}-\underline \delta-(\overline \delta-\underline \delta)/2)$.
    Consider the sup $\overline x_1$ of the points $x \leq x_{k_2}$ such
    that $c(x)=c(x_{k_1}-\underline \delta)$, and the inf $\overline x_2$
    of the points $x\geq x_{k_1}$ such that $c(x)=c(x_{k_2}+\underline \delta)$. 
    For  $\epsilon$ and $n$ large enough, $c_n$ must
    be constant outside of $[\overline x_1-\epsilon,
    \overline x_2 + \epsilon]$, and linear in some interval of 
    length $\underline \delta$ included in $[\overline x_1-\epsilon,
    \overline x_2 + \epsilon]$, connecting two values
    at distance at most $\epsilon$ from $c(\overline x_1)$
    and $c(\overline x_2)$, respectively. 
    Since this holds for all small $\epsilon$,
    it implies that $|\overline x_1-\overline x_2|=\underline \delta$,
    and that $c$ must be linear in between these points, connecting 
    $c(\overline x_1)$
    and $c(\overline x_2)$ by continuity.
    
\end{enumerate}
Pairs of points in which $(2)$ happens cannot overlap,
hence indicate by $\alpha_1, \alpha_2, \dots, \alpha_{k-1}$
be  the smallest points of each pair, arranged in increasing
order; 
let $\alpha_0=0, \alpha_k=T$;
and let $c_0\leq \beta_1, \dots, \beta_k\leq 1$
be the such that $c(x)=\beta_i$ for $x \in [\alpha_{i-1}+
\underline \delta,
\alpha_i]$ for $i=1, \dots, k$. We have shown that
$c(x)$ satisfies \eqref{piecelin} for these values
of $\alpha$'s and $\beta$'s.
This finishes the proof.

\end{proof}

Consider now the following minimization problem 
$$
\min_{c\in\mathcal{K}}\mathcal{L}(c)
$$
with $\mathcal{L}$ as in \eqref{loss}.
We now show that the functional is continuous in $c$: once this is proved, by Weierstrass Theorem we conclude that a global minimum $c*\in\mathcal{K}$ of the functional $\mathcal{L}$ exists.
To prove continuity we use the well-posedness of the S-E-A-I-R-D  model.
In fact, let $\Vec{X}=(S,E,A,I,R,D)$  and denote by $\Vec{F}(c,\Vec{X})$ the vector-valued function having as components the right-hand sides of the S-E-A-I-R-D differential equations.

Then we can rewrite the system in vector form
$$
\Vec{X}'=\Vec{F}(c,\Vec{X}),\,\, \Vec{X}(0)=\Vec{X}^0.
$$
where by assumption the norm\footnote{Here $\|\Vec{X}\|=\max_{1\leq i\leq6}\left\{\max_{[0,T]}|X_i(t)|\right\}$} of the solution $\vec{X}$ is such that $\|\Vec{X}\|\leq 1$ and $\Vec{F}$ is smooth in both variables.
Let now $c_n\in \mathcal{K}$ such that $c_n$ converges uniformly in $[0,T]$ to a function $c\in \mathcal{K}$.
Consider now the solution $\Vec{X}_n\in C^1[0,T]$ of
$$
\Vec{X}'=\Vec{F}(c_n,\Vec{X}),\,\, \Vec{X}(0)=\Vec{X}^0
$$
and denote by $\Vec{X}\in C^1[0,T]$ the solution to 
$$
\Vec{X}'=\Vec{F}(c,\Vec{X}),\,\, \Vec{X}(0)=\Vec{X}^0
$$
Then $\Vec{W}_n=\Vec{X}_n-\Vec{X}$ is solution to 
$$
\Vec{W}'_n=\Vec{F}(c_n,\Vec{X}_n)-\Vec{F}(c,\Vec{X}),\,\,\,\Vec{W}(0)=\Vec{0}
$$
Now observe that 
$$
\Vec{F}(c_n,\Vec{X}_n)-\Vec{F}(c,\Vec{X})=\Vec{F}(c_n,\Vec{X}_n)-
\Vec{F}(c,\Vec{X}_n)+\Vec{F}(c,\Vec{X}_n)-\Vec{F}(c,\Vec{X})
$$
and by the smoothness of $\Vec{F}$, the boundness of $\Vec{X}_n$ and $\Vec{X}$ and the linear dependence of $\Vec{F}$ on $c$ we have the following bounds
$$
\|\Vec{F}(c_n,\Vec{X}_n)-\Vec{F}(c,\Vec{X}_n)\|\leq C\|c_n-c\|
$$
and 
$$
\|\Vec{F}(c,\Vec{X}_n)-\Vec{F}(c,\Vec{X})\|\leq K \|\Vec{W}_n\|
$$
From these last two inequalities we get the differential inequality 
$$
\|\Vec{W}'_n\|\leq K\|\Vec{W}_n\|+C\|c_n-c\|,\,\,\Vec{W}(0)=\Vec{0}
$$
which implies
$$
\max_{[0,T]}\|\Vec{W}_n\|\leq C\max_{[0,T]}|c_n-c|e^{KT}
$$
Hence, 
$$\max_{[0,T]}\|\Vec{W}_n\|\rightarrow 0$$ 
as $n\rightarrow \infty$ i.e. $$\max_{[0,T]}\|\Vec{X}_n-\Vec{X}\|\rightarrow 0$$ 
as $n\rightarrow \infty$.
Finally, noting that 
$$
\mathcal{L}(c_n)=\int_0^Tf(t,c_n,\Vec{X}_n)dt
$$
and since $f$  is continuous in all variables ($a D'=\delta X_4$),  $\max_{[0,T]}\|\Vec{X}_n-\Vec{X}\|\rightarrow 0$ and $\max_{[0,T]}|c_n-c|\rightarrow 0$  we finally obtain
$$
\mathcal{L}(c_n)\rightarrow \mathcal{L}(c)
$$ 
as $n\rightarrow \infty$.
\begin{rem} Clearly, existence of a minimum of the functional can be derived  in the more general class of controls that are uniformly Lipschitz continuous in $[0,T]$ with values in $[c_0,1]$ again by compactness and continuity of $\mathcal{L}$. 
\end{rem}

\subsection{The first order optimality conditions} \label{pontriagyn}

We now derive the first order optimality conditions in the form of Pontryagin minimum principle, \cite{P}, for  the constrained optimization problem 
\begin{equation}\label{Min}
\min_{c\in\mathcal{K}}\mathcal{L}(c)=\min_{c\in\mathcal{K}} \intop_0^{T}e^{-rt}\left[
\frac{1-\left(c(t)^{\theta} \left[S+E+A+R\right]\right)^{1-\sigma}}{1-\sigma}
+aD'(t)\right]dt 
\end{equation}
under the constraint 
\begin{equation}\label{const}
\vec{X}'=\vec{F}(c,\Vec{X}),\,\,\vec{X}(0)=0.
\end{equation}
where $\mathcal{K}$ is the class of controls defined in the previous section.
Let $\Vec{X}^*$ and $c^*\in \mathcal{K}$ be the optimal pair for the above constrained minimization problem. 
%$$
%\mathcal{V}(P)=\frac{1-P^{1-\sigma}}{1-\sigma}
%$$ 
%where 
%$$
%P=c^{\theta}(S+E+A+R)
%$$

%$$
%\frac{\partial \mathcal{V}(P)}{\partial %c}=\mathcal{V}'(P)}P_c=-\theta P^{-\sigma} c^{\theta-1}(S+E+A+R)=
%-\theta c^{\theta(1-\sigma)-1}(S+E+A+R)^{1-\sigma}}
%$$
Then the augmented Hamiltonian is
$$
\mathcal{H}=e^{-rt}\left(\frac{1-\left(c^{\theta} \left[S+E+A+R\right]\right)^{1-\sigma}}{1-\sigma}+a\delta I+e^{rt}\vec{\lambda}\cdot\vec{F}+e^{rt}w_1(1-c)+e^{rt}w_2c\right)
$$
and considering now
$$
\mathcal{\tilde H}=e^{rt}\mathcal{H}=\frac{1-\left(c^{\theta} \left[S+E+A+R\right]\right)^{1-\sigma}}{1-\sigma}+a\delta I+e^{rt}\vec{\lambda}\cdot\vec{F}+e^{rt}w_1(1-c)+e^{rt}w_2c
$$
where $\vec{\lambda}=(\lambda_S,\lambda_E,\lambda_A,\lambda_I,\lambda_R,\lambda_D)$, $w_1$ and $w_2$ are two non-negative functions. Set $\vec{\mu}=e^{rt}\vec{\lambda}$ and $v_1=e^{rt}w_1$, $v_2=e^{rt}w_2$, then we can express the optimality conditions in terms of the Hamiltonian $\mathcal{\tilde H}$, i.e.,  
%Using now the fact that $c\in\mathcal{K}$ and hence the fact that $c$ is a piece-wise linear function of the form (\ref{piecelin})
%$$
%\mathcal{\tilde H}=G(c,\vec{X},\vec{\mu},w_1,w_2)
%$$
%where $\vec{c}=(\alpha_0,\dots,\alpha_k,\beta_0,\dots,\beta_k)$

$$
\mathcal{\tilde H}_c^*=0
$$
where $\mathcal{\tilde H}_c^*$ indicates the  derivative with respect to $c$ of $\mathcal{\tilde H}(c,\Vec{X}^*,\vec{\mu}^*,v^*)$, 
i.e. 
$$
-\theta c^{\theta(1-\sigma)-1}(S^*+E^*+A^*+R^*)^{1-\sigma}-\mu^*_S \beta S^*(sI^*+E^*+A^*)+\mu^*_E \beta S^*(sI^*+E^*+A^*)-v^*_1+v^*_2=0
$$
where $v^*_1,v^*_2\geq 0$ and the vector $\vec{X}^*$ and $\vec{\mu^*}$ are respectively the solution of the direct problem and of the adjoint linear problem along the optimal solution $c=c_*(t)$, that is

\begin{align*}
  \left\{\begin{array}{rcl}
&\mu'_S-r\mu_S&=c_*^{(1-\sigma)\theta}(S^*+E^*+A^*+R^*)^{-\sigma}+\mu_S( n+\beta c_*(sI^*+E^*+A^*))
-\mu_E\beta c_*(sI^*+E^*+A^*)\\
&\mu'_E-r\mu_E&=c_*^{(1-\sigma)\theta}(S^*+E^*+A^*+R^*)^{-\sigma}+\mu_S\beta c^*S^*-\mu_E(\beta c^*S^*-(\kappa+n))-\mu_A(1-\epsilon)\kappa-\mu_I\kappa\epsilon\\
&\mu'_A-r\mu_A&=c_*^{(1-\sigma)\theta}(S^*+E^*+A^*+R^*)^{-\sigma}+\mu_S\beta c_*S^*-\mu_E\beta c_*S^*+\mu_A(\gamma+n)-\mu_R\gamma\\
&\mu'_I-r\mu_I&=-a\delta+\mu_S\beta s c_*S^*-\mu_E\beta s c_*S^*+\mu_I(\gamma+\delta+n)-\mu_R\gamma-\mu_D\delta\\
&\mu'_R-r\mu_R&=c_*^{(1-\sigma)\theta}(S^*+E^*+A^*+R^*)^{-\sigma}+\mu_Rn\\
&\mu'_D-r\mu_D&=\mu_Sn-\mu_D\delta,\\
&\vec{\mu}(T)&=0
\end{array}
\right.
\end{align*}

One can use the optimality conditions to compute the optimal control in a larger class of functions and use it as benchmark for the suboptimal control that we find in the class $\mathcal{K}$.

%We now derive the first order optimality conditions in the form of Pontryagin minimum principle for the retarded control optimization problem (RCOP)
%$$
%\min_{c\in\mathcal{K}}\mathcal{L}(c)=\min_{c\in\mathcal{K}}\int_0^Tf(t,c(t),c(t-\tau),\Vec{X})dt
%$$
%subject to the constraint
%$$
%\Vec{X}'=\Vec{F}(c(t),\Vec{X}(t)),t\in [0,T]\,\, \Vec{X}(0)=\Vec{X}^0
%$$
%This kind of problem has been extensively studied and we refer for example to the following paper (http://www.aimsciences.org/article/doi/10.3934/mbe.2018051).
%The idea in \cite{} is to reduce the delayed control problem to a non delayed problem to which Pontryagin’s minimum principles applicable.
%In fact the Hamiltonian $H$ and the augmented Hamiltonian $\mathcal{H}$ for the delayed control problem are defined similarly as in the nondelayed case introducing a new control variable $v$ denoting the delayed control variable.
%$$
%H(t,c,v,\Vec{X},\lambda)=f(t,c,v, \Vec{X})+\lambda\cdot\Vec{F}(c,\Vec{X})
%$$
%$$
%\mathcal{H}(t,c,v, \Vec{X},\Vec{\lambda})=f(t,c,v, \Vec{X})+\vec{\lambda}\cdot\Vec{F}(c,\Vec{X})+\mu(1-c)
%$$
%If $c^*$ and $\Vec{X}^*$ denote the optimal control and the corresponding optimal state of the retarded problem, then there exists an adjoint state $\vec{\lambda}^*$ and a multiplier function $\mu^*$ 
%satisfying respectively
%$$
%\Vec{\lambda}'=\mathcal{H}_{\Vec{X}}(t,c^*(t),c^*(t-\tau),\Vec{X}^*,\Vec{\lambda}),\,\, \vec{\lambda}^*(T)=0
%$$
%\smallskip
\begin{eqnarray*}
%&H^*(t)+\chi_{[0,T-\tau]}H^*(t+\tau)\leq\\ &H(t,c(t),c^*(t-\tau),\Vec{X}^*(t),\Vec{\lambda}^*)+ \chi_{[0,T-\tau]}H(t+\tau,c(t+\tau),c(t), \Vec{X}^*(t+\tau),\Vec{\lambda}^*(t+\tau))
\end{eqnarray*}
%\smallskip
%
%$$\mathcal{H}^*_c(t)+\chi_{[0,T-\tau]}\mathcal{H}^*_v(t+\tau)=0
%$$
%and finally 
%$$
%\mu^*(t)\geq 0,\,\, \mu^*(t)(1-c^*(t))=0
%$$
%where $\mathcal{H}^*(t)$ and $\mathcal{H}^*_{\Vec{X}}(t)$  represent respectively the value and the derivative of $\mathcal{H}^*$  at the point $(t,c^*(t),c^*(t-\tau),\Vec{X}^*(t),\Vec{\lambda}^*(t),\mu^*(t))$.

\subsection{On uniqueness of the optimal control} \label{uniqueness}
The functional $\mathcal{L}$ in  \eqref{loss} is in general not
convex, and there are  no reasons to expect uniqueness of the
optimal control in $\mathcal{K}$. In fact, in some cases
the cost functional appears to undergo a phase transition
 in the social cost of COVID-19 death $a$.
 Typically, real valued functions of systems undergoing a phase transition 
 are  convex in one phase and 
 concave in the other (see, e.g.,  the percolation probability as
 function of its intensity parameter, \cite{G}, Figure 2.3),
 which is a further justification for the observed loss
 of convexity of $\mathcal{L}$.
 In addition, at the critical value of $a$  multiple optimal 
 controls can appear. 
 
 In the simple case of a unique, long term lockdown
 imposed at Day $85$ to an opening level $\overline c$, and by a suitable choice of the parameters within
 the realistic ranges described below 
  in Section \ref{calibration}, one
  can numerically find a value of $a$ for which there are
  two minimizers of $\mathcal{L}$. 
  
 A graph of $\mathcal{L}$ is plotted in Figure \ref{twominima} 
  as function of $\overline c$. 
  At the selected value of $a$, an optimal strategy is to exert no lockdown, 
  but another optimal solution is to impose
  an opening level $\overline c=84$. The two solutions have
  different overall mortality and GDP loss, but the same
  value of the loss functional, hence they are
  equivalent for the social planner, and for all
  those agreeing with her/his parameter selection and
  perceived social cost of a COVID-19 death.

 \vspace{0.4cm}
\begin{figure}[h!] \label{twominima}
    \centering
    \includegraphics[scale=0.35]{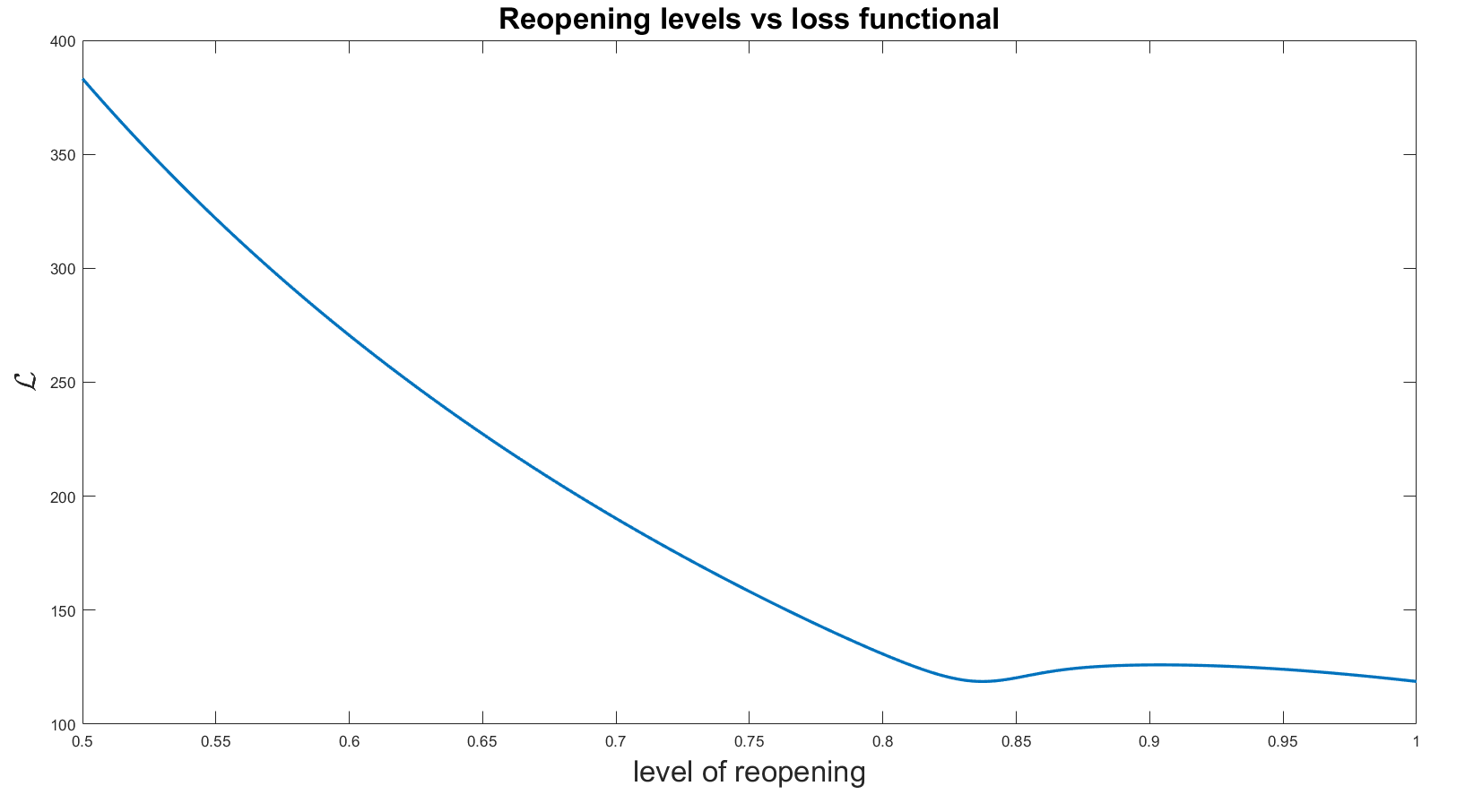}
    \caption{An example of two minima. 
    $\mathcal{L}$ as function of the reopening
    intensity $\overline c$ applied from day $85$ to $460$.
    The social cost of Covid death is fixed at $a\approx 7833,11$ and $r=0.00001$. See Section \ref{calibration} for the other parameters. }
    \label{two-min}
\end{figure}

\vspace{0.4cm}

At  values of $a$ which seem to better reflect current valuations, that is for a higher value of the social cost
of Covid deaths,
see Section \ref{calibration}, the minimum is likely to occur in the
phase in which $D(t)$ is also convex, which is at lower values of
$\overline c$, and therefore it
is unique. This is the case in all the examples of
the next section.

\section{Examples of optimal policies.} \label{Examples}
 \subsection{Parameter selection}\label{calibration}
 %beta=0.25
%s=0.1
%n=0.00003
%kappa=0.2 --> 5??
%epsilon=2/3  ---> 2/3 develop symptoms??
%gamma=0.14 --> 7
%delta=0.0028

There is a large variability in
the estimations of the COVID-19 infection rate $\beta$ \cite{AT};
we adopt the average value of
$\beta\approx 0.25$.
The reduced exposure to infected individuals who have 
developed symptoms is difficult to estimate: 
we start from a factor of $s=0.1$ and carry out a sensitivity analysis.
Duration of the latency period after infection
and before symptoms are developed has been 
estimated in about $5$ days (see for example \cite{QL} and \cite{K-W}), so that $\kappa \approx 0.2$.
The fraction of asymptomatic is also quite problematic,
with estimates ranging from $5\%$ to $60\%$; 
we take an average value of $(1-\epsilon)=1/3$, 
estimated in one of the studies, \cite{N}.
Similarly, the average recovery period is about
$7$ days, for mild cases
\cite{Byrne}, suggesting $\gamma \approx 0.14$
for the recovery rate of an asymptomatic;
in general, more severe cases worsen after
about $7$ days, requiring hospitalization, which completely excludes
them from the possibility of transmission: for this
reason, we also use the same value of $\gamma \approx 0.14$
for moving these cases from the infected to recovered,
where most of them eventually will be;
one fraction eventually dies, with the rate discussed now.
The death to recovery rate is a highly controversial 
value, as both the recorded number of infected
and deaths are affected by error which could range to 
$1000\%$. We take $\delta/\gamma \approx 0.02$ in such a
way that the overall mortality rate in the population
if the epidemics spreads without control ends up being
about $1\%$; this is in line with several studies and observations:
\cite{AB} estimates a US mortality of $1.3\%$;
 the Institute Pasteur indicates  $0.53\%$ \cite{salje2020estimating};
   and several locations have observed an increase of
   overall mortality up to six-fold \cite{ISTAT}; this is compatible
with a COVID-19 death rate of about $1\%$ 
spread over the two months very likely needed for the uncontrolled 
virus to infect everyone in a limited area. Finally, the natural mortality rate is taken to be $3 \times 10^{-5}$  corresponding to about 12
death per year per $1000$, 
which is an average natural mortality  rate in industrialized countries.\footnote{https://data.worldbank.org/indicator/SP.DYN.CDRT.IN\color{black}}

With these assumptions, the equations become
\begin{align} 
\frac{dS}{dt} & =-0.25 \hskip 1mm S  \hskip 1mm
c(t)(0.1 \hskip 1mm I+E+A)-0.00003 S + 0.00003(1-D)\label{eq:dS1}\\
\frac{dE}{dt} & =0.25 S(0.1 I+E+A)-(0.2+0.00003)E\label{eq:dE1}\\
\frac{dA}{dt} & = 0.2/3 E -(0.14+
0.00003) A \label{eq:dA}\\
\frac{dI}{dt} & =0.4/3 E-(0.14+
0.00283)I\label{eq:dI1}\\
\frac{dR}{dt} & =0.14 (A+I){-0.00003 R}\label{eq:dR1}\\
\frac{dD}{dt} & =0.0028 I\label{eq:dD1}\\
\frac{dD_N}{dt} & = 0.00003(S+E+A+I+R)\label{eq:dDN1}
\end{align}
As we take as initial time a very early stage of the epidemic
outbreak (for all countries except China), we assume that the number of initial exposed is very small, in the order of
one in a million; hence we take $S(0)=1-10^{-6}, E(0)=10^{-6},  A(0)=I(0)=R(0)=D(0)=0$.
A more accurate model, taking care of the geographical dispersion
of the population would include different contact rates
for individual living in far away areas \cite{MarinoGattoetal}

As a verification of parameter selection, we show
that the mortality reproduces current observations,
see Figure \ref{total_deaths1}. 
Figure \ref{total_deaths2} illustrates the risk of a
restart of the outbreak after the first reopening.

\begin{figure}[h!]
    \centering
    \includegraphics[scale=0.3]{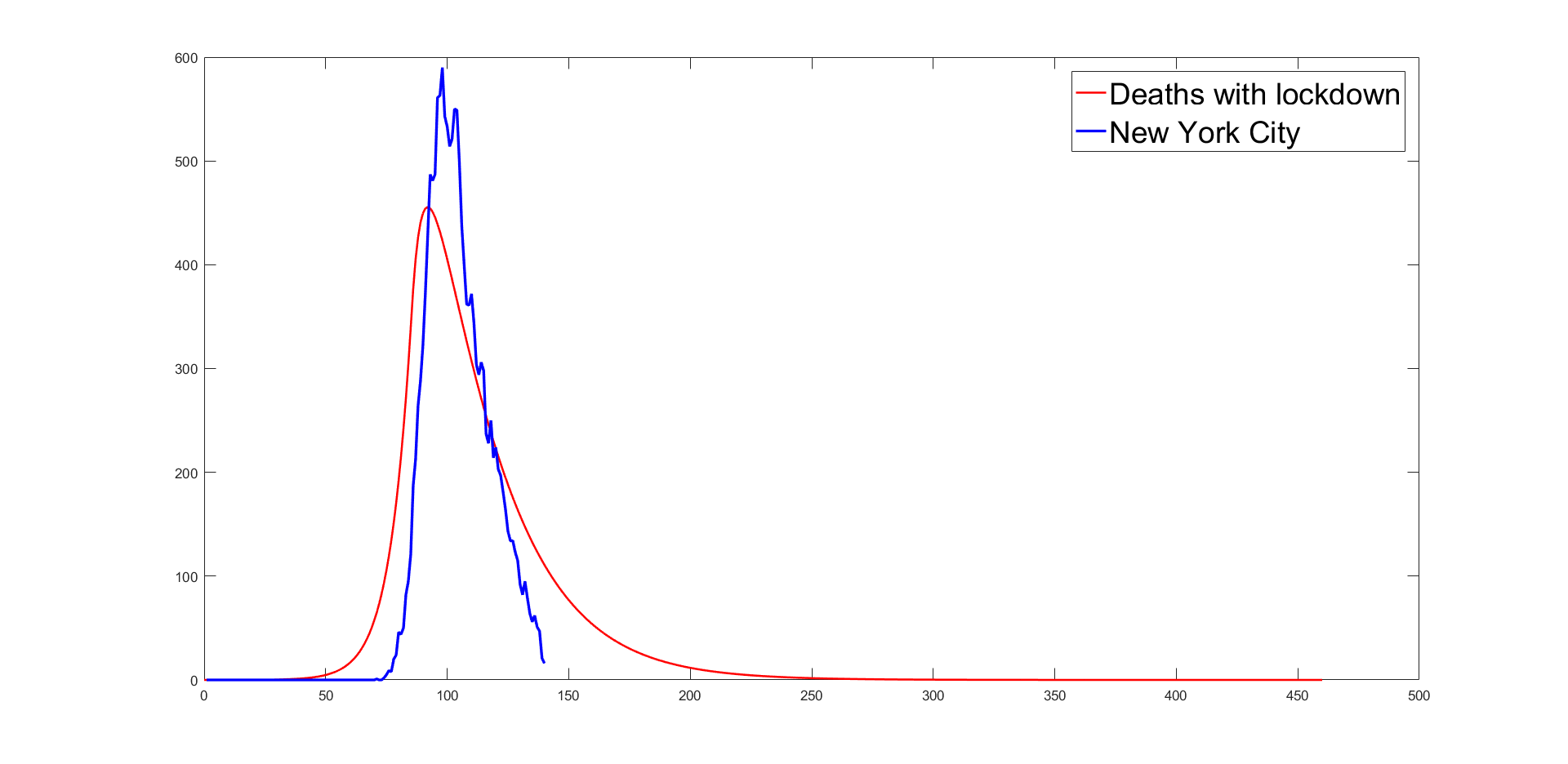}
    \caption{ Covid deaths for an outbreak followed by lockdown at Day 85.
    Comparison is
    with real data of NYC. Notice that NYC seems to have 
    a slightly higher transmission rate $\beta$,
    and has imposed a stricter containment policy than
    the one assumed by the graph of the mortality 
    in our model. \footnotesize{Note: the model is calibrated to fit cities or equivalent homogeneous areas and does not represent an entire country.}}
    \label{total_deaths1}
\end{figure}

\begin{figure}[h!]
    \centering
    \includegraphics[scale=0.3]{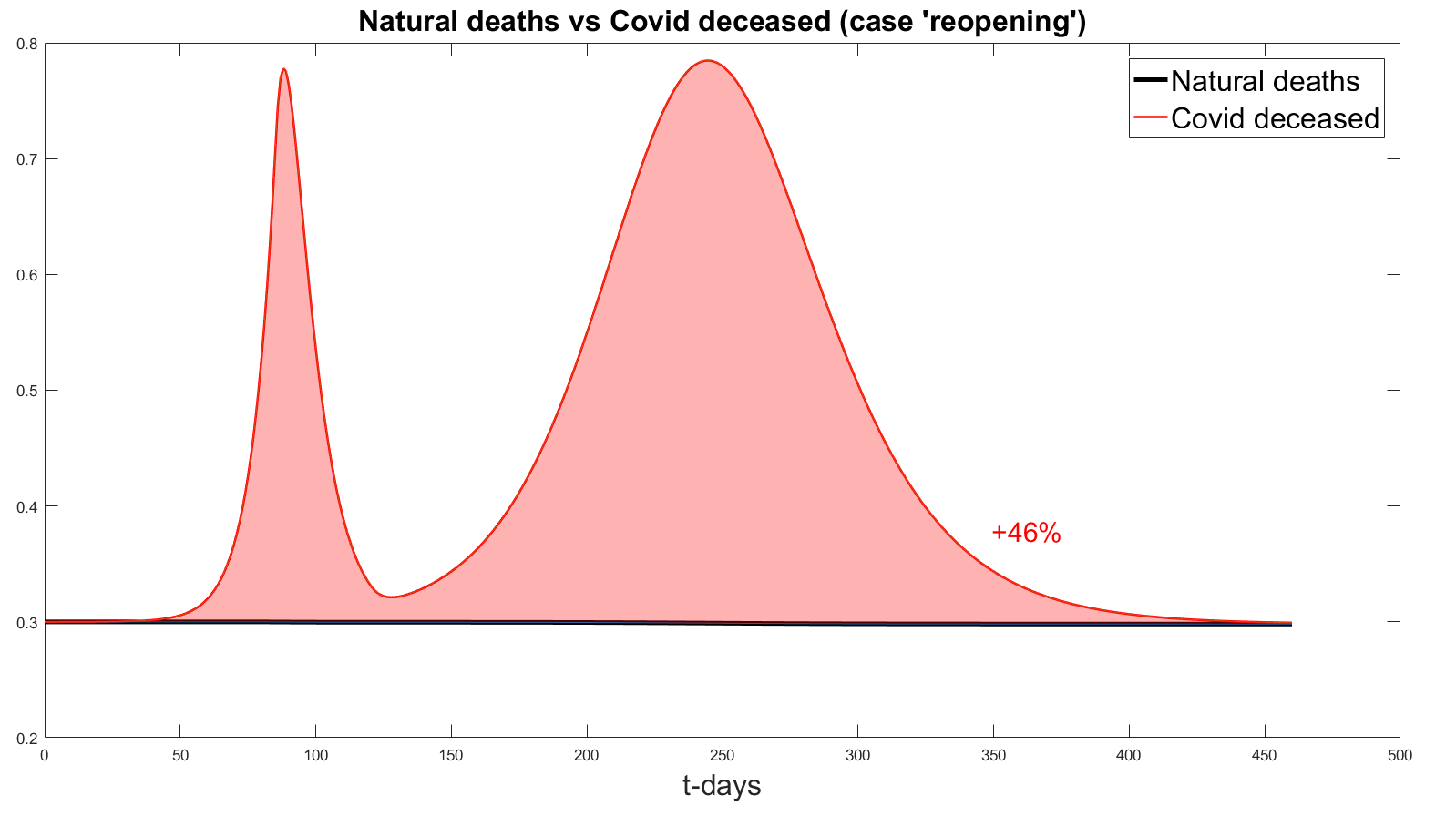}
    \caption{Total deaths for an outbreak, followed by lockdown at Day 85, and reopening at Day 120.
    The percentage represents the increase of deaths with respect to the natural ones. \footnotesize{Note: the model is calibrated to fit cities or equivalent homogeneous areas and does not represent an entire country.}}
    \label{total_deaths2}
\end{figure}
%\vskip 2cm

%\vskip 1cm

The yearly discount rate $r$ in various developed countries \color{black}  is currently in the range $-0.75$ to $5.5\%$; we assume a discount rate of $4\%$ but we check the impact of a wide range of alternative assumptions in the sensitivity analysis.
The exponent $\sigma$ of the function $\mathcal{V}$
is taken to be $\sigma=2$, leading to an intertemporal elasticity of substitution within the year of 1/2.

The elasticity parameter $\theta$ needs to be considered carefully.
To estimate it, we recall that the reproduction number
 \eqref{R0} has been estimated in various countries 
 before and after a lockdown, see Table \ref{tab:ThetaEstimates}. From \eqref{production},
 at each point in time
\[
\log P=\theta\left(\log c(t)\right)+\log(S+E+A+R)
\]
so that, considering two times,
$t^-$ shortly before, and $t^+$ shortly after
a lockdown, we have
\[
\log\frac{P(t^{+})}{P(t^{-})} \approx \theta\left(\log\frac{c(t^{+})}{c(t^{-})}\right) \approx \theta
\log\frac{\mathcal{R}(t^{+})}{\mathcal{R}(t^{-})},
\]
where in the first approximation,
we neglected the variation in the potential labor force $S+E+A+R$, sinc,e between $t^{-}$ and $t^{+}$, 
the labor force available for production is  assumed to be only impacted by the variations in $c$; the second
approximation follows
from  \eqref{R0} again
neglecting variations in $S(t)$ in the short interval.
This gives the estimate
\begin{equation}\label{eq:theta_calc}
\theta \approx\log\frac{P(t^{+})}{P(t^{-})}/\log\frac{\mathcal{R}(t^{+})}{\mathcal{R}(t^{-})}.
\end{equation}

\begin{table}[htbp]
\centering {\footnotesize{}\caption{Alternative values of $\theta$, from various studies and variants.}\label{tab:ThetaEstimates}
}%
\begin{tabular}{cccccccccc}
\toprule 
\multicolumn{4}{c}{} &  &  &  &  &  & \tabularnewline
\midrule

\midrule 
\hline\hline
\multicolumn{4}{c}{} &  &  &  & \tabularnewline
 & {\footnotesize{}GDP loss} & {\footnotesize{}$\log\frac{P(t^{+})}{P(t^{-})}$} & {\footnotesize{}Source} & {\footnotesize{}$\mathcal{R}(t)$} & {\footnotesize{}$\log\frac{\mathcal{R}(t^{+})}{\mathcal{R}(t^{-})}$} &  & {\footnotesize{}Source} & \textbf{{\footnotesize{}Implied $\theta$}} & \tabularnewline
 & {\footnotesize{}(Instantan. } &  &  &  &  &  &  &  & \tabularnewline
{\footnotesize{}Country/Region} & {\footnotesize{}or monthly) } &  &  &  &  &  &  &  & \tabularnewline
\hline
\midrule 
{\footnotesize{}France} & {\footnotesize{}-36\%} & {\footnotesize{}-0.405} & {\footnotesize{}(A)} & {\footnotesize{}From 3 to 1} & {\footnotesize{}-1.099} &  & {\footnotesize{}(0)} & \textbf{{\footnotesize{}0.369}} & \tabularnewline
{\footnotesize{}France (2)} & {\footnotesize{}-} & {\footnotesize{}-} & {\footnotesize{}-} & {\footnotesize{}From 3 to 0.5} & {\footnotesize{}-1.792} &  & {\footnotesize{}(b)} & \textbf{{\footnotesize{}0.226}} & \tabularnewline

{\footnotesize{}France (3)} & {\footnotesize{}-} & {\footnotesize{}-} & {\footnotesize{}-} & {\footnotesize{}From 3.15 to 0.27} & {\footnotesize{}-2.457} &  & {\footnotesize{}(a)} & \textbf{{\footnotesize{}0.165}} & \tabularnewline
\hline
\midrule
{\footnotesize{}Italy} & {\footnotesize{}-36\%} & {\footnotesize{}-0.405} & {\footnotesize{}(B)} & {\footnotesize{}From 3.54 to 0.19} & {\footnotesize{}-2.925} &  & {\footnotesize{}(a)} & \textbf{{\footnotesize{}0.139}} & \tabularnewline
\hline
\midrule
{\footnotesize{}Germany} & {\footnotesize{}-30\%} & {\footnotesize{}-0.357} & {\footnotesize{}(B)}{\footnotesize{} } & {\footnotesize{}From 3 to 1} & {\footnotesize{}-1.099} &  & {\footnotesize{}(c)} & \textbf{{\footnotesize{}0.325}} & \tabularnewline
{\footnotesize{}Germany (2)} &  &  &  & {\footnotesize{}From 3.34 to 0.52} & {\footnotesize{}-1.860} &  & {\footnotesize{}(a)} & \textbf{{\footnotesize{}0.192}} & \tabularnewline
\hline
\midrule
{\footnotesize{}Sweden} & {\footnotesize{}-20\%} & {\footnotesize{}-0.223} & {\footnotesize{}(B)}{\footnotesize{} } & {\footnotesize{}From 3.04 to 2.02} & {\footnotesize{}-0.409} &  & {\footnotesize{}(a)} & \textbf{{\footnotesize{}0.545}} & \tabularnewline
\hline
\midrule
{\footnotesize{}US (late March)} & {\footnotesize{}-10.0\%} & {\footnotesize{}-0.105} & {\footnotesize{}(C)}{\footnotesize{} } & {\footnotesize{}From 1.50 (to 1)} & {\footnotesize{}-0.405} &  & {\footnotesize{}(d)} & \textbf{{\footnotesize{}0.260}} & \tabularnewline
{\footnotesize{}US (2) (late March)} & - & - & - & {\footnotesize{}From 2,20 (to 1)} & {\footnotesize{}-0.788} &  & {\footnotesize{}(e)} & \textbf{{\footnotesize{}0.134}} & \tabularnewline
{\footnotesize{}US (3) (late March)} & {\footnotesize{}-10\%} & {\footnotesize{}-} & - & {\footnotesize{}From 2 to 1} & {\footnotesize{}-0.693} &  & {\footnotesize{}(f)} & \textbf{{\footnotesize{}0.152}} & \tabularnewline
\hline
\midrule
{\footnotesize{}US (4) (May)} & {\footnotesize{}-31.0\%} & {\footnotesize{}-0.371} & {\footnotesize{}(C)}{\footnotesize{} } & {\footnotesize{}From 3 to 1 } & {\footnotesize{}-1.099} &  & {\footnotesize{}(0)} & \textbf{{\footnotesize{}0.338}} & \tabularnewline
{\footnotesize{}US (5) (May)} & {\footnotesize{}-34.9\%} & {\footnotesize{}-0.430} & {\footnotesize{}(D)}{\footnotesize{} } & {\footnotesize{}From 3 to 1} & {\footnotesize{}-1.099} &  & {\footnotesize{}(0)} & \textbf{{\footnotesize{}0.391}} & \tabularnewline
\hline\hline

{\footnotesize{}Our preferred benchmark} & {\footnotesize{}-23.3\%} & {\footnotesize{}-0.265} & {\footnotesize{}-} & {\footnotesize{}From 2 to 0.8} & {\footnotesize{}-0.916} &   & {\footnotesize{}(*)} & \textbf{{\footnotesize{}1/3 }} & \tabularnewline
\midrule
\hline\hline
\multicolumn{10}{l}{{\tiny{}Notes: specification and sources.}}\tabularnewline
\multicolumn{10}{l}{{\tiny{}(0): Priors; (*): our simulated benchmark outcome; }  {\tiny{}(a): }{\tiny{}\cite{bryant2020estimating}}; {\tiny{}(b): }{\tiny{}\cite{dimdorel}}}
\tabularnewline
\multicolumn{10}{l}{{\tiny{}(c): }{\tiny{}\cite{hamouda2020schatzung}}; {\tiny{}(d): \cite{eichenbaum2020macroeconomics} }; {\tiny{}(e): }{\tiny{}\cite{riou2020pattern}}   }\tabularnewline
\multicolumn{10}{l}{{\tiny{}(f): }{\tiny{}\cite{jones2020optimal}}; {\tiny{}(A) INSEE, April 2020, Point conjoncture}}\tabularnewline
\multicolumn{10}{l}{{\tiny{}(B) OECD Nowcasts, Coronavirus: The world
economy in freefall, http://www.oecd.org/economy/}}\tabularnewline
\multicolumn{10}{l}{{\tiny{}(C) Fed Atlanta GDPNow tracker (8/10/2020)}}\tabularnewline
\multicolumn{10}{l}{{\tiny{}(D) New York Fed Staff Nowcast https://www.forexlive.com/centralbank/!/the-ny-fed-nowcast-tracks-2q-growth-at-3122-20200508 }}\tabularnewline
\multicolumn{10}{l}{{\tiny{}(E) Sweden: Forecast for 2020 are estimated to
be between -6.9\% and 9.7\% by Statistics Sweden and the Riskbank,
approx. 2/3rd of the decline in France. }}\tabularnewline
\multicolumn{10}{l}{{\tiny{}https://www.cnbc.com/2020/04/30/coronavirus-sweden-economy-to-contract-as-severely-as-the-rest-of-europe.html}}\tabularnewline
\hline
\end{tabular}
\end{table}

\vspace{0.4cm}

Table \ref{tab:ThetaEstimates} shows various examples of co-variations of $\mathcal{R}$ and instantaneous GDP variations estimating from now-casting studies from various economic and statistical institutions after the lockdown from various countries. The parameters displayed have different sources. Some come from estimates based on data, other are simulated from epidemiologic models, and some are used in calibrations in economic papers, as a way to compare ourselves to the previous studies. The variability in the value of $\theta$ in the table is due to this diversity of methods. The range is between $0.166$ and $1.142$, with an average of $0.27$ and a s.d. of $0.12$. We select a value of $1/3$ that can be adapted to any country or period as indicated in the table.\footnote{A careful reader might notice that in the last row, the variation of the reproduction number and our best GDP response correspond to a value of $\theta=0.290$, slightly below our parameter choice $(1/3)$, the difference being due to the approximation in Equation \ref{eq:theta_calc}.}

In order to identify the time horizon of our analysis, we make several assumptions
about the evolution of the epidemic. In particular, we
assume that the policy assessment can be made with a specific
time frame in mind, after which technological advancements
like a therapy or a vaccine  will drastically reduce
the negative effects of the infection:  \cite{HHS}
and \cite{TT} predict a vaccine in early 2021,
and challenge trials will anticipate things even further.
 We then assume a prototypical
situation in which 
the epidemic has started unobserved in January 2020,
and we assume that it 
 will resolve at the end of the first quarter of 2021,
 hence we
take $T=460$ days.
Clearly, these periods are only indicative, and 
one can adapt the time frame when more reliable perspectives
are identifiable.

\vskip 1cm

The choice of the social cost $a$ of a Covid death is
particularly complex, as it depends on a variety of
socio-political and economic factors. 
We take a value of $a\approx 10,000$.
To assess a value of $a$, note that it implies from Table \ref{comparison2_experiment1} a decline of GDP of 76.7\% from day 85 to day 460, that is a loss of yearly GDP equal to $\frac{460-85}{365}\times 0.767=0.7886$, that is, a $21.2\%$ decline in yearly GDP. The gain is a decline in mortality of 0.74\%. If these numbers where applied to the case of France, with a GDP of 2778 billions USD in 2018 and a population of 67 million, each live saved would correspond to 1.758 million USD.  This is smaller than the statistical value of life currently estimated in developed economies, that is closer to 3 million euros \cite{baumstark2013elements} but one has to remember that most of the fatality have been for older individuals. According to various statistical sources \cite{Stat2020}, only 10\% of the deaths were aged below 65, while 71\% were aged above 75. We also report in Appendix Table \ref{tab:Fatality_Rates_Age} the fatality rates by age as available from recent studies. This implies that the right value for the statistical value of life in the exercise has to be lower than the usual estimates
\cite{Zenioetal}. Another factor is that the government lockdown was based on lower estimations for the proportion of deaths. The current range is large, going from 0.4\% for symptomatic according to the CDC or 0.37\% per infected in the so-called Gemeinde Gangelt study in Germany \cite{streeck2020vorlaufiges} to more than 4\%.  \color{black}

%Another way to compare our $a$ to real numbers is to use the logic of revealed preferences. The French government implemented in March until May 11, a 55 days of lockdown that will have approximately reduced GDP by 10\% and expect a decline in the theoretical number of deaths from $1\%$ of the population, 670 000, to only 50 000. This implies 620 000 lives saved, that would be evaluated as 5 to 10\% of yearly GDP. This implies 450 000 USD for a 10\% loss of GDP. To explain these relatively low values, again one can either involve the age distribution of deaths, or the fact that the government acted as if it believed in a lower fatality rate, or both. For instance, the range in the fatality rate is large, and has been found to be between 0.37\% in the so-called Gemeinde Gangelt study in Germany, \cite{streeck2020vorlaufiges} to more than 4\%. With a 0.37\% fatality, the value of each life saved from the revealed preferences of the lockdown policy would have been 1.4 million USD, close to our number.

Note that taking into account the risk aversion of the loss function does not change significantly the numbers involved and the order of magnitudes are preserved: risk-aversion mostly affect the numbers as $(0.9)^2$ that is by 20\% only. To see this, consider a small time interval of length $\Delta t=1$, so that the loss function is $\left(\mathcal{V}(P)\Delta t+aD\right)$ where $D$ is the number deaths over that interval: differentiating the expression along the iso-loss curve, the slope of the iso-loss (indifference) curve is exactly: 

\[
\frac{dP}{dD}=\frac{a}{-\mathcal{V}'(P)}=\frac{a}{P^{-\sigma}}=aP^{\sigma}
\]
 hence the adjustment factor is of the order of magnitude of the fraction of loss of  GDP $P$ to the square.

It is seen in the examples below that this value of the
social cost of Covid death corresponds to prefer a substantial
mortality reduction over GDP preservation, a phenomenon that, although
sporadically opposed by some political groups, has found
substantial support in most industrialized countries
\cite{Hart}.
Such value of $a$ is large enough that the optimal control functions
determine an effective containment of the spread of the
virus; this implies that the minimum of $\mathcal L$ occurs
where the total mortality is also likely to be convex as
function of the control, and
that the minimum is likely to be unique (see Section
\ref{uniqueness}). 

\vskip 0.5cm

We analyze below several examples of containment:
 \begin{itemize}
 \item A first policy is a containment with opening level 
 $\overline c$ until the end of the study period.
 \item A second policy is a containment  with opening level  $\underline c$
 till day $120$,  followed by a higher opening level 
 $\overline c$ until the end of the study period.
 \item A third policy is to implement several cycles of alternated higher and lower opening .
 \end{itemize}

As the presence of the virus went substantially unnoticed
in the early stages in most locations, and then some
time we needed to pass the required legislation,
we 
assume that all lockdowns begin on day 85; this
corresponds to March 25. Lock down in most countries,
except China,
started between March 9 and April 23, with a median
on March $25$\footnote{https://en.wikipedia.org/wiki/National\_responses\_to\_the\_COVID-19\_pandemic}.
When considering reopening, we use Day $120$,
which corresponds to 
April 29. For countries which have substantially reduced
containment measures as of May $5$th, the median
end date of lockdown has been April $24$, with about $20$
countries still in lockdown.

All the numerical examples below are computed by Matlab R2016, using   discretized ordinary differential equations (``ode45'' or ``ode23tb'' functions) and integrals.

\subsection{Optimal unique lockdown} \label{SectUniqLock}
We consider in this section a unique lockdown measure
imposed on Day 85 (March 25):
the opening level is reduced at level $\overline c$, and 
these restrictions are kept in place for 
the entire period, which is till Day 460, April 4, 2021.
While this could have been a viable policy,
implementing a moderate containment,
the extent of the resulting 
GDP loss turns out to be dramatic.
Figure \ref{onelock}
compares production reduction and mortality for the 
various levels of $\overline c$.

\begin{figure}[h!]
    \centering
    \includegraphics[scale=0.35]{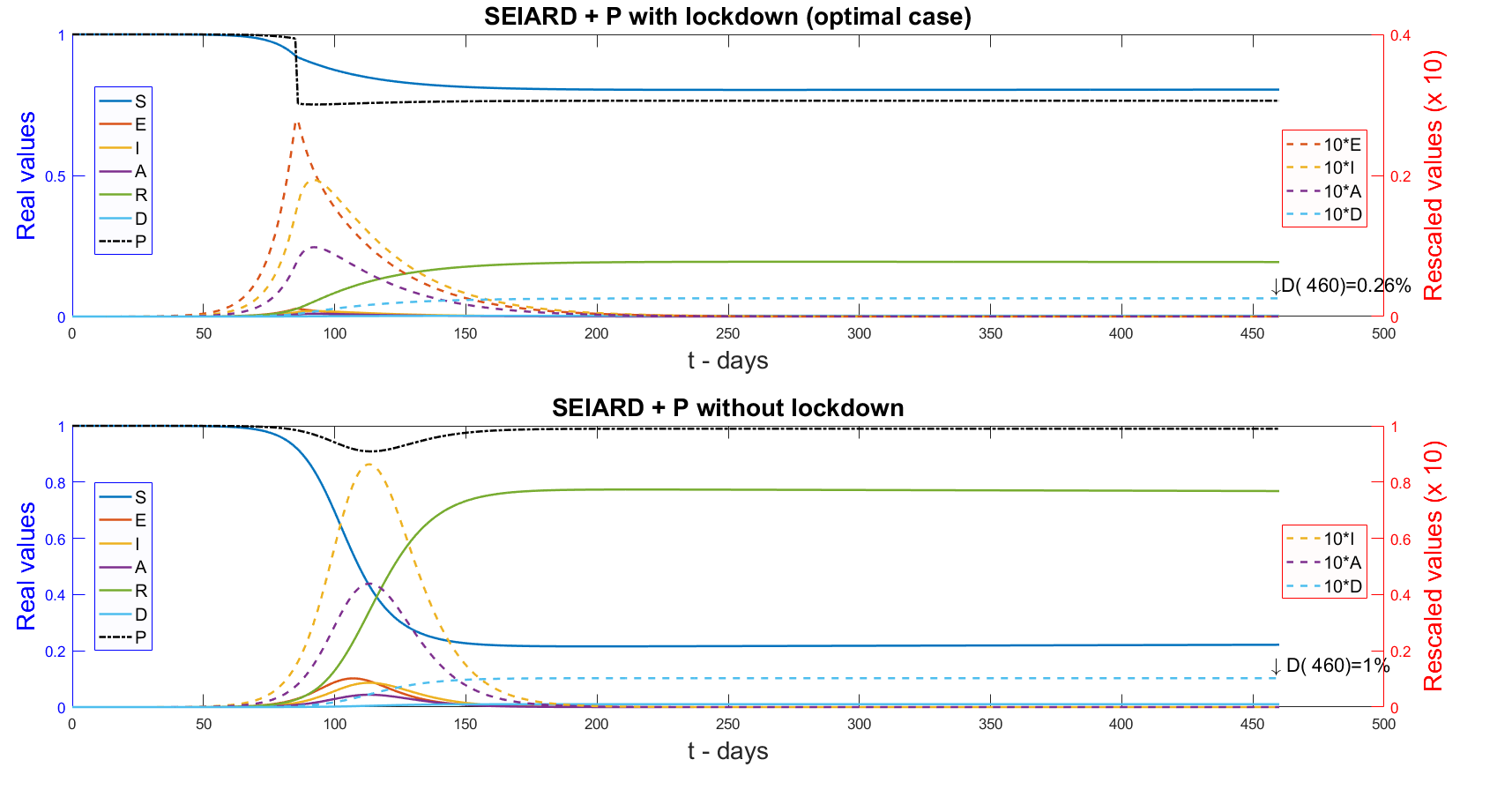}
    \caption{Comparison between the optimal moblity level and the case of no restrictions.}
    \label{plot_opt_lockdown}
\end{figure}
The optimal opening level is numerically determined to be 
$\overline c=76.7\%$. Figure \ref{plot_opt_lockdown}
compares the optimal containment policy with the case of 
no containment; Figure
\ref{comparison2_experiment1}, compares the optimal case
with two different policies corresponding to  less or more reduced opening  levels.

\begin{figure}[h!]
    \centering
    \includegraphics[scale=0.35]{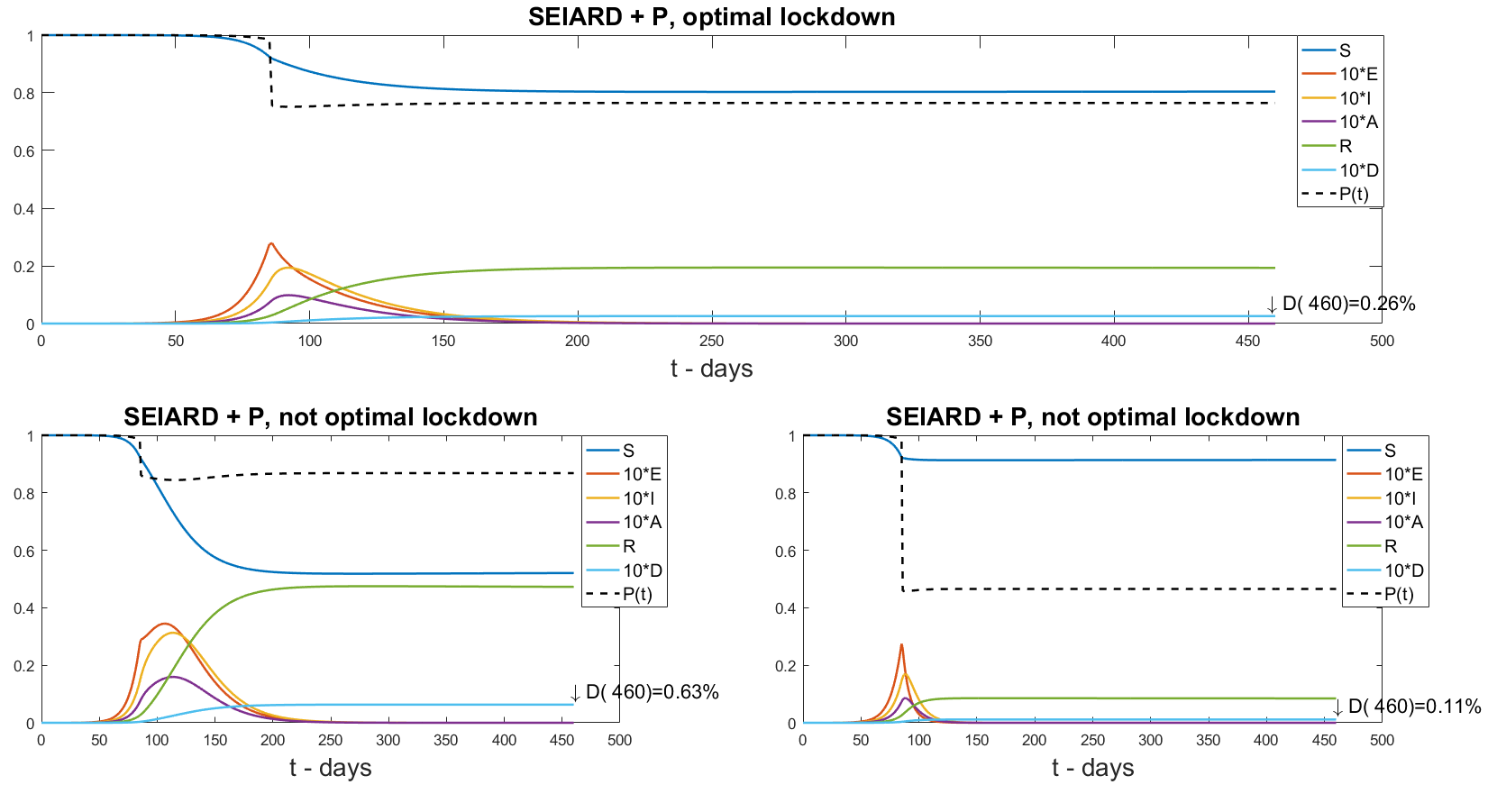}
    \caption{Fig. \ref{comparison2_experiment1}.A Top: optimal opening level and related epidemic variables. Fig. \ref{comparison2_experiment1}.B, Bottom left: a opening level higher than optimal. Fig.
    \ref{comparison2_experiment1}.C, Bottom right: a opening
    level below optimal.}
    \label{comparison2_experiment1}
\end{figure}

\begin{table}[htbp]
  \centering
  \caption{Single extended lockdown.
  }
  
    \begin{tabular}{lcccc}
     
 & \multicolumn{1}{c}{Epidemic} & \multicolumn{1}{c}{ 
        Insufficient restrictions} & \multicolumn{1}{c}{Optimal}& \multicolumn{1}{c}{Excessive restrictions} \\    & \multicolumn{1}{c}{No policy } & \multicolumn{1}{c}{ 
        Fig. \ref{comparison2_experiment1}.B } & \multicolumn{1}{c}{Fig. \ref{comparison2_experiment1}.A \%} & \multicolumn{1}{c}{Fig. \ref{comparison2_experiment1}.C } \\ \hline \vspace{0.1cm}
     Containment and opening level $\bar{c}$ & 100\% & 87.4\% & 76.6\% & 46.6\% \\  \vspace{0.1cm}    
   Mortality at Day $85$ & 0.03\%  & 0.03\%  & 0.03\%  & 0.03\% \\ \vspace{0.1cm}
    Total mortality at Day $460$ & 1.03\%  & 0.63\%  & 0.26\%  & 0.11\% \\  \vspace{0.1cm}
    Total mortality reduction & 0\%   & 38.47\% & 74.85\% & 88.96\% \\ \vspace{0.1cm}
    Annualized 1st quarter GDP loss   & 2.42\%  & 2.89\%  & 3.29\%  & 4.43\% \\ \vspace{0.1cm}
    Total annualized GDP loss  & 1.78\% & 11.28\%  & 19.45\% & 43.67\% \\ \vspace{0.1cm}
    Value loss functional & 129.53 & 88.45 & 75.82 & 130.59\\ \hline 
    
    \end{tabular}%
  \label{tab:onlylockdown}%
\end{table}%

Notice that in case of no restrictions, the total mortality
is about $1\%$, and  annualized GDP loss 
due to passage of the virus is $1.78\%$.
As noted in the Introduction, the lockdown realizes a sharp containment of
mortality, but the constraint of protracted measures
causes a dramatic GDP loss.

Figure \ref{reproduction_exp1} compares the time evolution of the reproduction numbers in the cases of optimal lockdown and no lockdown: notice that the optimal lockdown quickly brings
the reproduction number to slightly below $1$, keeping it there for the
entire period.

\begin{figure}[h!]
    \centering
    \includegraphics[scale=0.3]{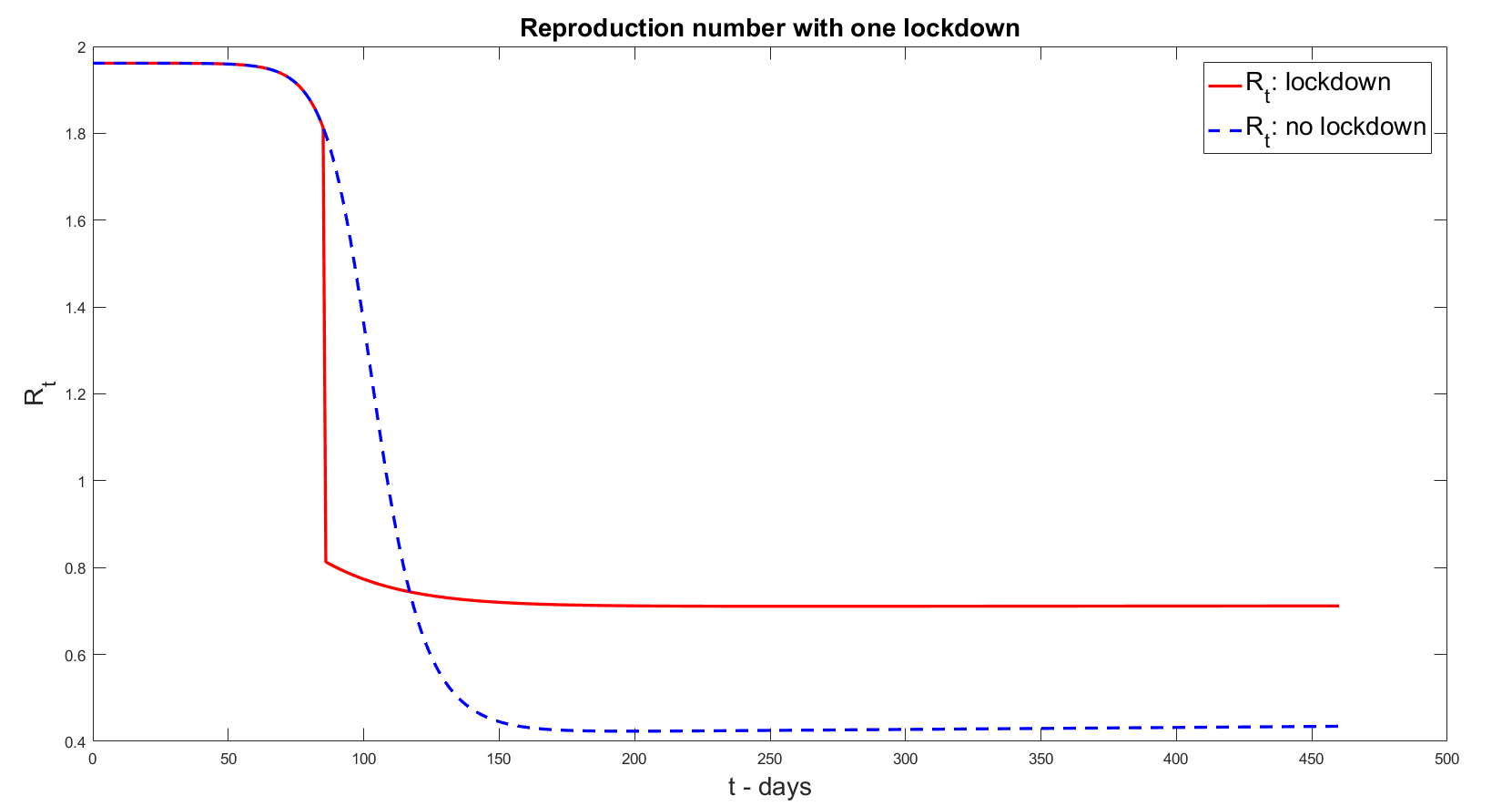}
    \caption{Reproduction number in the case of the optimal unique extended lockdown and without lockdown.}
    \label{reproduction_exp1}
\end{figure}

\vskip 2cm

\subsection{Optimal reopening level}
\label{sec:opt_reopening}
Most countries have imposed severe restrictions
after a first period, which is at Day $85$ in
our model, followed by a sizeable reopening after 
about two months. To simulate this situation,
we assume that at Day $85$
the  opening level has been fixed at $\underline c=0.5$; as the previous example shows,
this would not be optimal if imposed for a long
time, and it incorporates the assumption
of a release after a relative short period. In  accordance to the current reopening in many countries,
 the containment is
relaxed  to level $\overline c$
at Day $120$. 
Clearly, in this case a loss of production has already been
incurred because of the initial containment, and 
we have selected an opening level that  reproduces 
the observed loss of GDP in the first quarter at
an annual rate of $4$-$5\%$, see Table \ref{tab:onlyreopening}, Line 5.

We then numerically determine the optimal level of reopening,
which turns out to be
at   $\overline c \approx 90.1\%$.
Figure  \ref{comparison2_exp2} compares the optimal solution
with non-optimal ones, and a detailed comparison
of some of the outcomes is carried out in Table \ref{tab:onlyreopening}.

\begin{figure}[h!]
    \centering
    \includegraphics[scale=0.35]{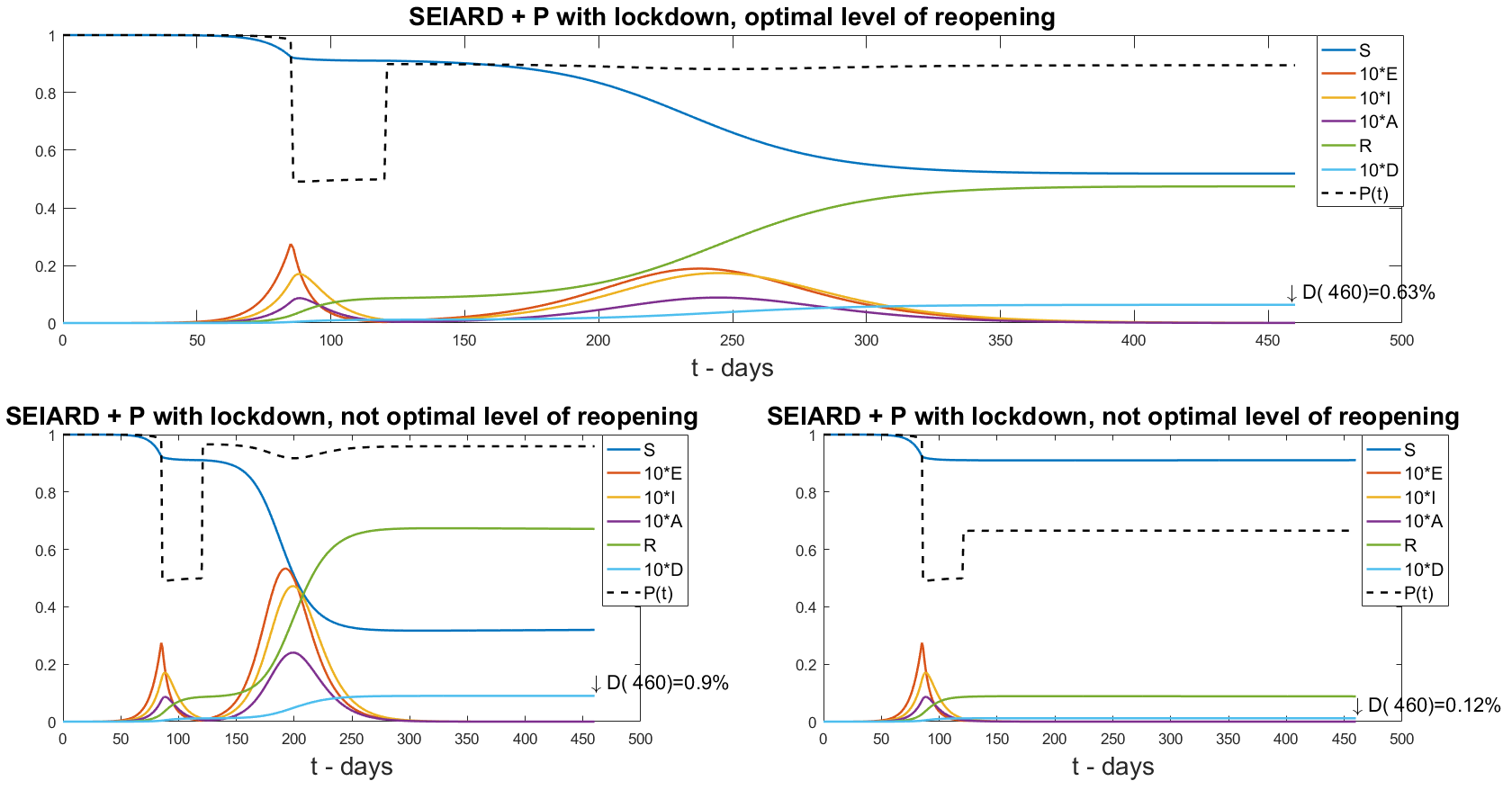}
    \caption{Fig. \ref{comparison2_exp2}.A, Top: optimal reopening level. Fig. \ref{comparison2_exp2}.B, Bottom left: an excessive reopening. Fig. \ref{comparison2_exp2}.C, Bottom right: a suboptimal reopening level.}
    \label{comparison2_exp2}
\end{figure}

\begin{table}[htbp]
  \centering
  \caption{One reopening after a lockdown.}
%  \centering
%  \caption{Some numerical results for the second experiment dealing with the optimal reopening after a drastic lockdown. Columns two, three and four are the results related to the charts in Figure \ref{comparison2_exp2}.}
    \begin{tabular}{lcccc}
         & \multicolumn{1}{c}{Epidemic} & \multicolumn{1}{c}{ 
        High reopening} & \multicolumn{1}{c}{Opt. reopening}& \multicolumn{1}{c}{Limited reopening} \\
          & \multicolumn{1}{c}{No policy } & \multicolumn{1}{c}{ 
        Fig. \ref{comparison2_exp2}.B } & \multicolumn{1}{c}{Fig. \ref{comparison2_exp2}.A } & \multicolumn{1}{c}{Fig. \ref{comparison2_exp2}.C } \\ \hline \vspace{0.1cm}
    Reopening level $\bar{c}$ & 100\% & 96.8\% & 90.1\% & 66\% \\ \vspace{0.1cm}     
    Mortality at Day $85$ & 0.03\%  & 0.03\%  & 0.03\%  & 0.03\% \\ \vspace{0.1cm}
    Total mortality at Day $460$ & 1.03\%  & 0.90\%  & 0.63\%  & 0.12\% \\ \vspace{0.1cm}
    Mortality reduction & 0\%  & 12.78 & 38.43 & 88.54 \\  \vspace{0.1cm}
    Annualized 1st quarter GDP loss & 2.42\%  & 4.30\%  & 4.30\%  & 4.30\% \\  \vspace{0.1cm}
     Total annualized GDP loss & 1.78\%  & 7.53\% & 12.02\% & 28.72\% \\  \vspace{0.1cm}
     Value loss functional & 129.53 & 103.49 & 102.17 & 109.5\\ \hline
    \end{tabular}%
  \label{tab:onlyreopening}%
\end{table}%

Notice that the optimal reopening level achieves a 
substantial herd immunity by the so called "flattening the curve".
Because of that, the mortality reduction reaches $38.43\%$
only, with a more moderate, but still sizeable, annualized GDP loss of
$12.02\%$. Observe that deviations from optimality are
extremely ineffective.

The reproduction number in  Figure \ref{reproduction_exp2}, after drastically decreasing
and then going back higher, finally stabilizes around $0.8$.
\begin{figure}[h!]
    \centering
    \includegraphics[scale=0.3]{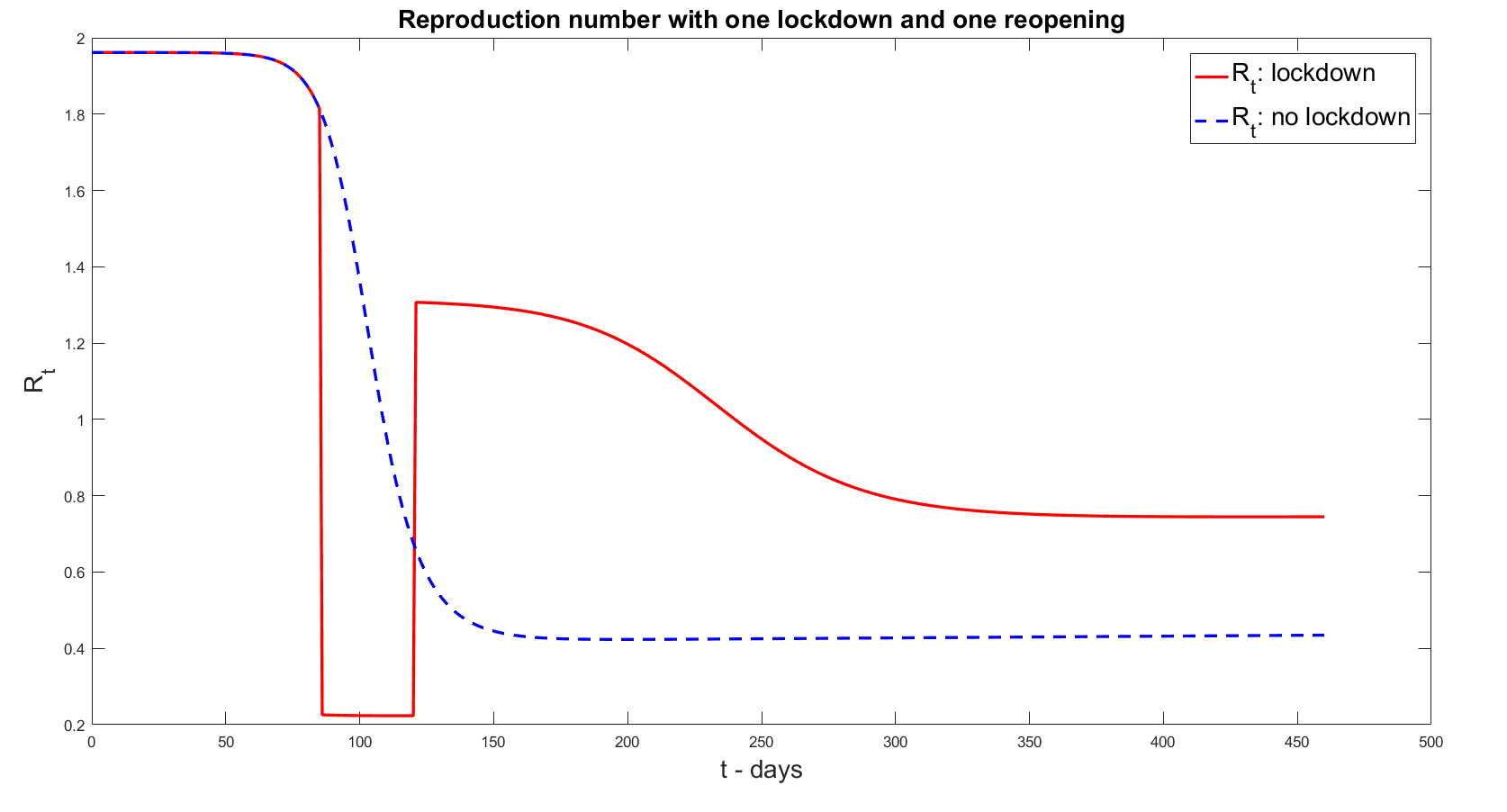}
    \caption{Reproduction number in the case of the optimal reopening.}
    \label{reproduction_exp2}
\end{figure}
%Finally, we show a chart on the total deaths, which include both natural and Covid deaths, and the percentage increase of death on the whole period, regarding a population of 10000 inhabitants, see Figure \ref{total_deaths}.

\newpage

\subsection{Optimal periodic containment}
\label{sec:opt_periodic_cont}
In this section we show some numerical results related to a periodic containment. We assume that after a drastic lockdown at an opening level of  $\underline c=0.5$, there is a
complete reopening, followed by two more lockdowns
at an opening level $\overline c$: we
optimize over $\overline c$, see Figure \ref{comparison2_exp3_case1}. Production loss vs. mortality
is plotted in Figure \ref{GDP_mortality_exp3_case1};
notice the peculiar effect of  too sharp lockdowns
when these are reapplied at Days $170$ and $230$: because of excessive containment,
the outbreak restarts later and
the mortality ends up being higher even with more GDP loss
than with the optimal control.
A third lockdown would be necessary in this case.

\begin{figure}[h!]
    \centering
    \includegraphics[scale=0.35]{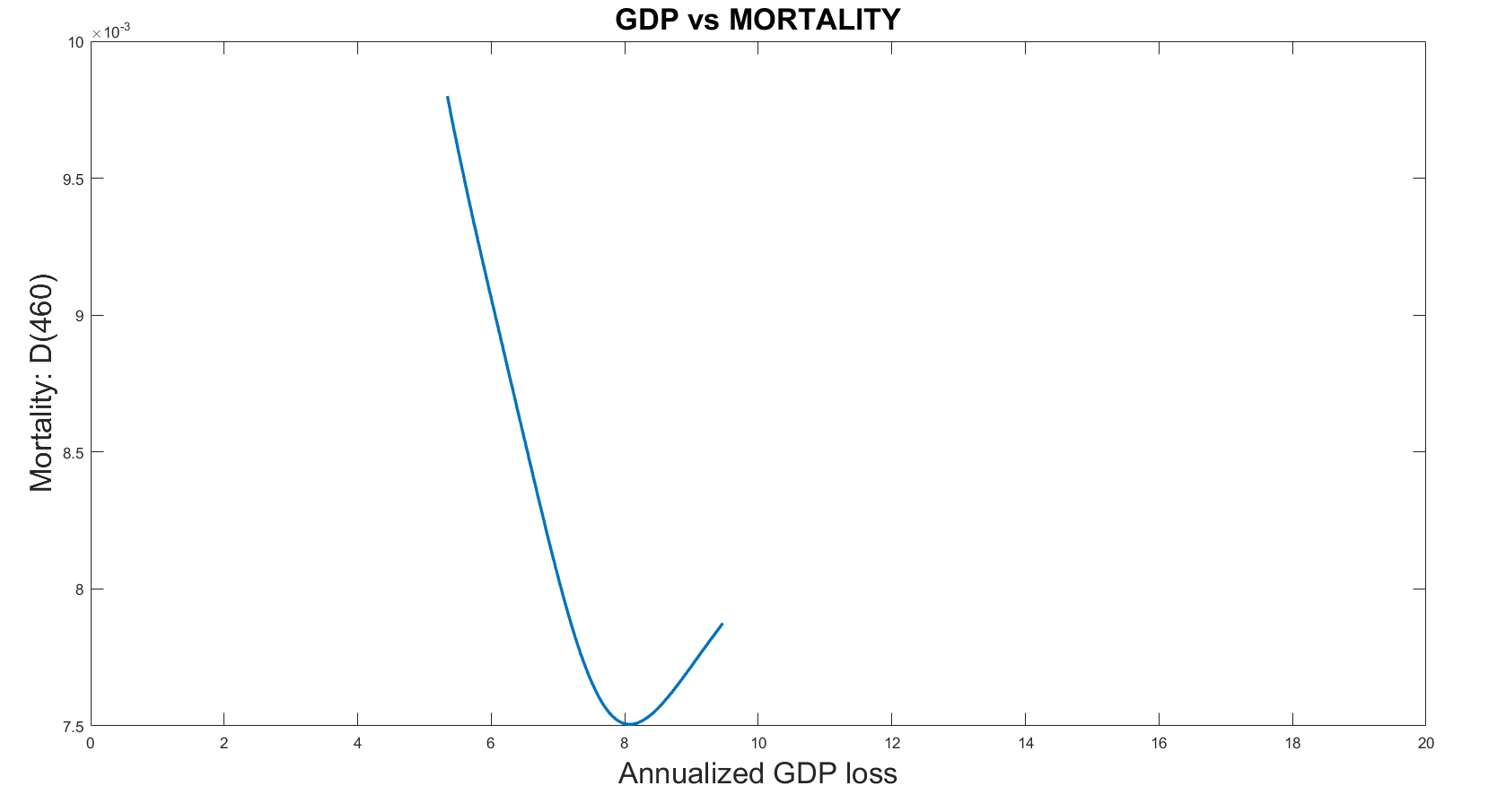}
    \caption{Production loss  and fraction of deaths   in  a periodic containment; curve is parametrized by the
    opening level during the two  containments
    that follow a first, fixed one.}
    \label{GDP_mortality_exp3_case1}
\end{figure}
The optimal opening level turns out to be $\overline c=72.9\%$.
This solution provides a moderate reduction in mortality
with a contained economic damage at an annualized
GDP loss of $7.44\%$, see Table
\ref{tab:exp3_case1}. Herd immunity is reached 
with a very low number of infected at the time, 
which is shown in  Figure \ref{reproduction_exp3_case1}
to be approximately Day $250$, when the 
reproduction number is finally set to just below $1$.
This is the ideal strategy to achieve herd
immunity, as discussed in \cite{BM},
and it has been automatically identified by the
 optimization process.

\begin{figure}[h!]
    \centering
    \includegraphics[scale=0.35]{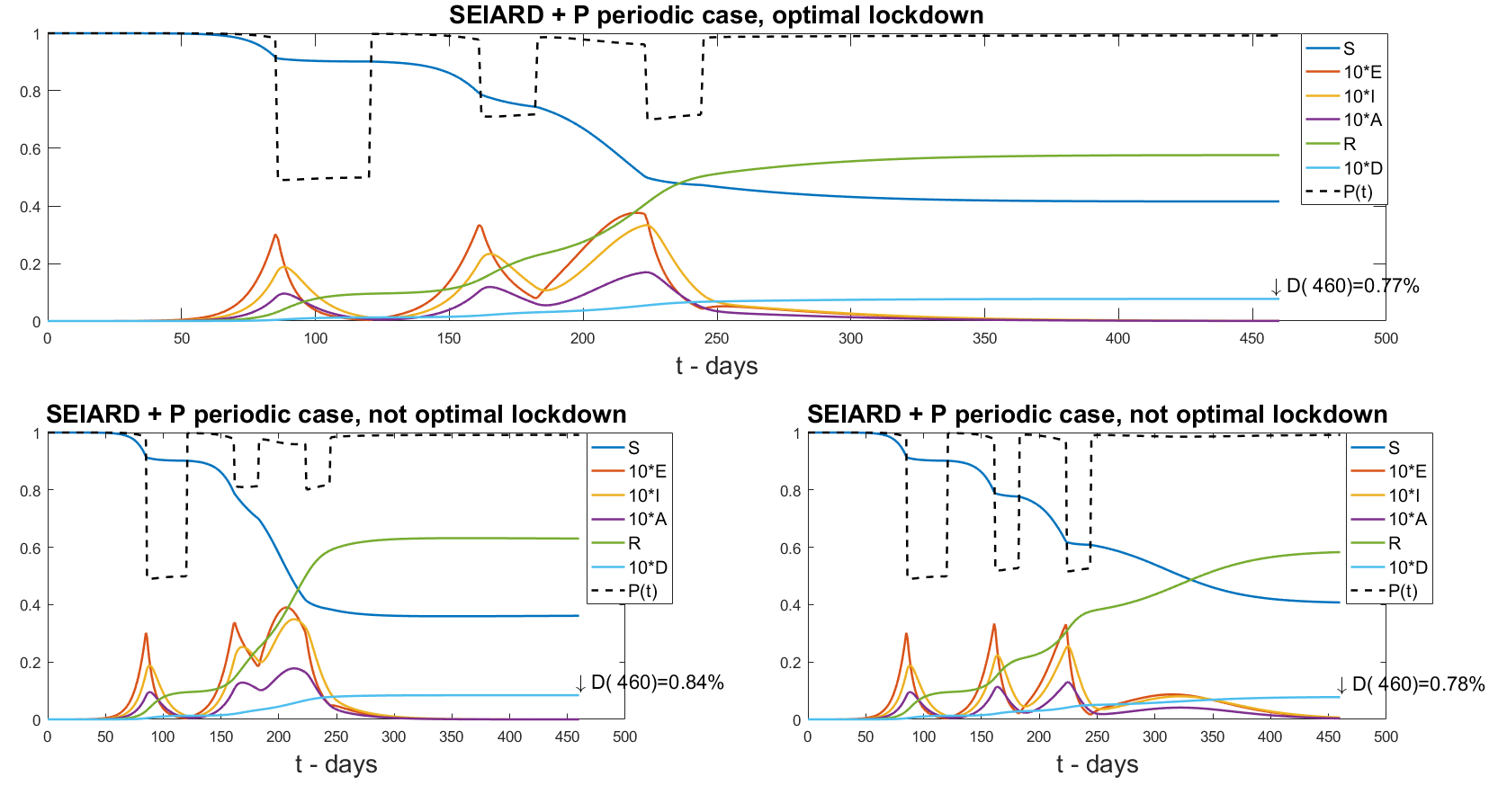}
    \caption{Fig. \ref{comparison2_exp3_case1}.A Top: optimal opening level. Fig. \ref{comparison2_exp3_case1}.B Bottom left: excessive opening. Fig. \ref{comparison2_exp3_case1}. C Bottom right: 
    excessively reduced opening.}
    \label{comparison2_exp3_case1}
\end{figure}

\begin{table}[htbp]
   \centering
  \caption{Periodic containment optimized over the
    opening level during the two containments 
    that follow a first, fixed one}
%  \centering
%  \caption{Some numerical results of a periodic containment, when we optimize the level starting from the second lockdown. The numerical results provided here comes from the results of charts in Figure \ref{comparison2_exp3_case1}.}
    \begin{tabular}{lcccc}
     & \multicolumn{1}{c}{Epidemic} & \multicolumn{1}{c}{ 
        Low lockdown} & \multicolumn{1}{c}{Opt. lockdown}& \multicolumn{1}{c}{Stricter lockdown} \\
          & \multicolumn{1}{c}{No policy } & \multicolumn{1}{c}{Fig. \ref{comparison2_exp3_case1}.B } & \multicolumn{1}{c}{Fig. \ref{comparison2_exp3_case1}.A } & \multicolumn{1}{c}{Fig. \ref{comparison2_exp3_case1}.C } \\ \hline \vspace{0.1cm}
    Second and third reopening level $\bar{c}$ & 100\% & 83.3\% & 72.9\% & 
    53.2\% \\  \vspace{0.1cm}   
   Mortality at Day $85$ & 0.03\%  & 0.03\%  & 0.03\%  & 0.03\% \\ \vspace{0.1cm}
    Total mortality at Day $460$ & 1.03\%  & 0.84\%  & 0.77\%  & 0.78\% \\ \vspace{0.1cm}
    Mortality reduction & 0\%  & 18.13\% & 25.15\% & 24.36\% \\ \vspace{0.1cm}
    Annualized 1st quarter GDP loss & 2.44\%  & 4.32\%  & 4.32\%  & 4.32\% \\ \vspace{0.1cm} 
    Total annualized GDP loss & 1.78\%  & 6.62\%  & 7.44\%  & 9.17\% \\ \vspace{0.1cm}   
     Value loss functional & 129.53 & 107.61 & 107.41 & 108.05\\ \hline    
    \end{tabular}%
  \label{tab:exp3_case1}%
\end{table}%

\begin{figure}[h!]
    \centering
    \includegraphics[scale=0.3]{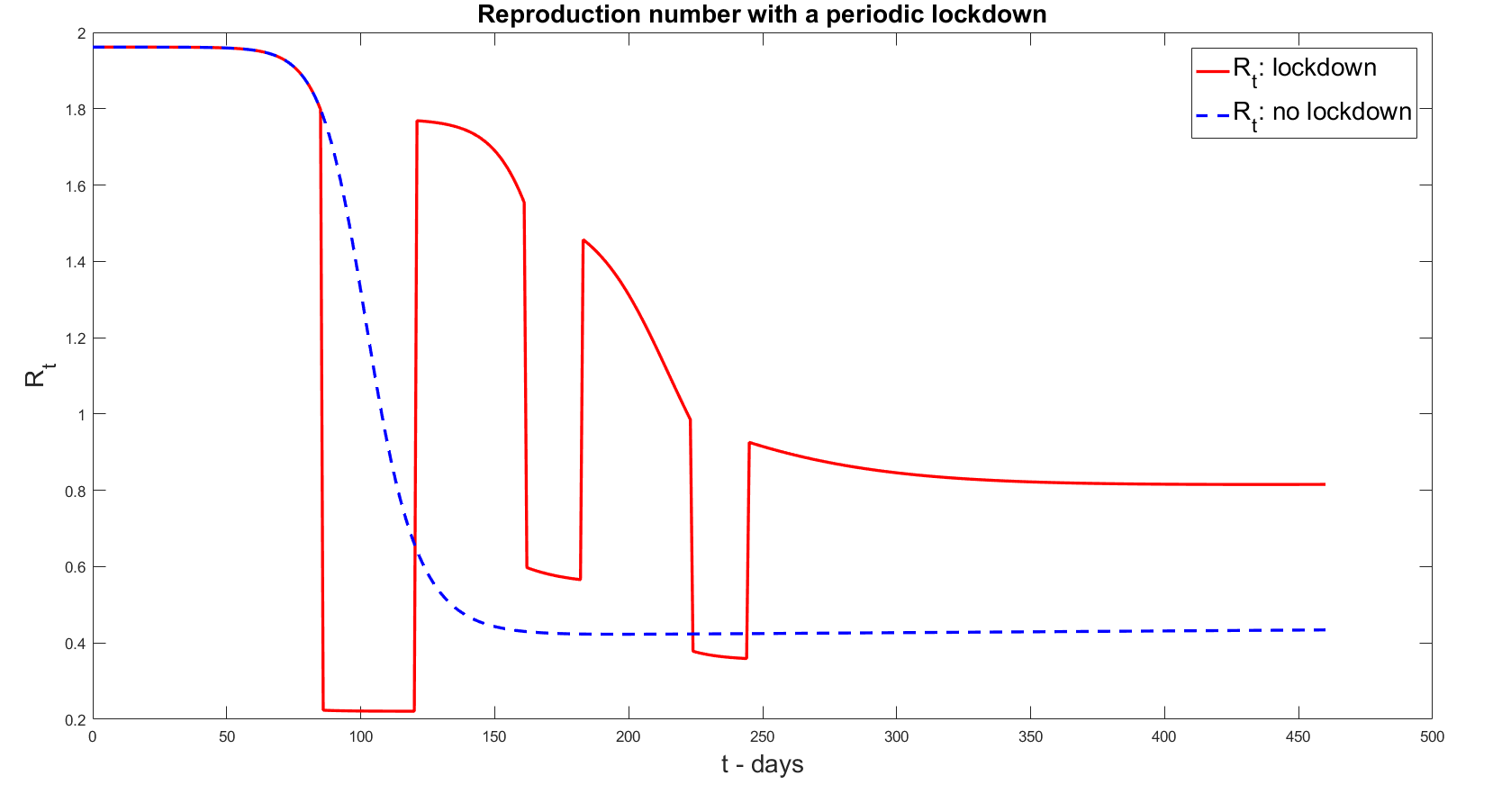}
    \caption{Reproduction number when
    control is optimized over the
    opening level during the two containments 
    that follow a first, fixed one.}
    \label{reproduction_exp3_case1}
\end{figure}

%%%%%%%%%%%%%%%%%%%%%%%

\subsection{Optimization over three parameters}
\label{sec:opt_periodic_cont}
In this example we optimize over three parameters:
after a first, fixed containment from Day $85$ to 
Day $120$ at opening level $\underline c=0.5$, two
more containment periods take place, at 
opening level $\overline c_1$, each for a time length $\overline \tau$, 
interspersed with reopening at level $ \overline c_2$: we optimize
over $\overline c_1,  \overline c_2 $, and $\overline \tau$.
A comparison of the optimal solution with others is in 
Figure  \ref{comparison2_exp4}; a summary is in 
Table \ref{tab:exp3_case_3_par}; and the
reproduction number is plotted in 
\ref{reproduction_exp4}.

\begin{figure}[h!]
    \centering
    \includegraphics[scale=0.35]{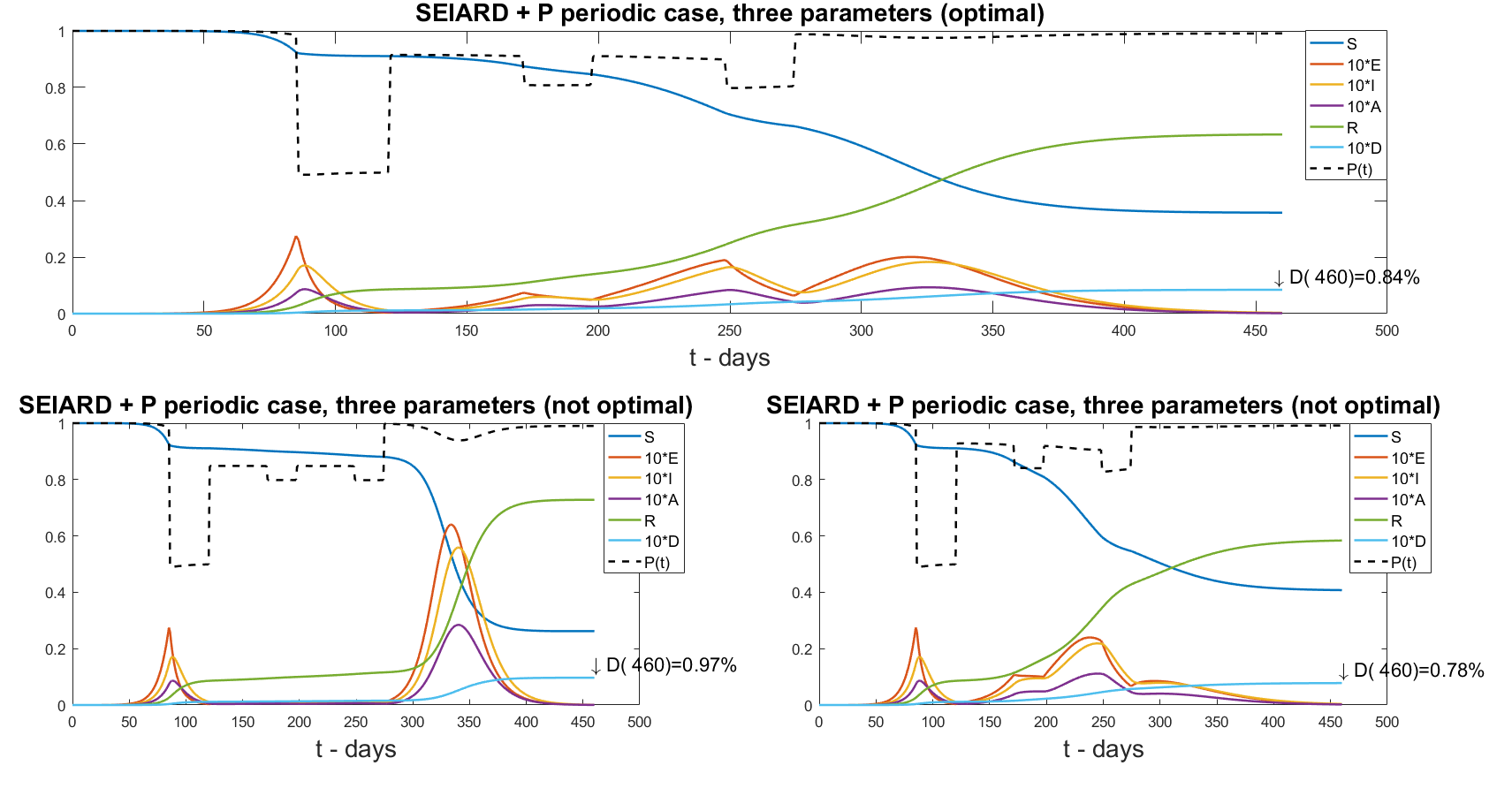}
    \caption{Fig. \ref{comparison2_exp4}.A Top: Optimal policy in the case of three parameters: opening level at containments
    after the first, fixed one, duration of containments,
    and level of in-between reopening. Fig. \ref{comparison2_exp4}.B Bottom left: a non optimal policy. Fig. \ref{comparison2_exp4}. C Bottom right: excessive reopening.}
    \label{comparison2_exp4}
\end{figure}

\begin{figure}[h!]
    \centering
    \includegraphics[scale=0.35]{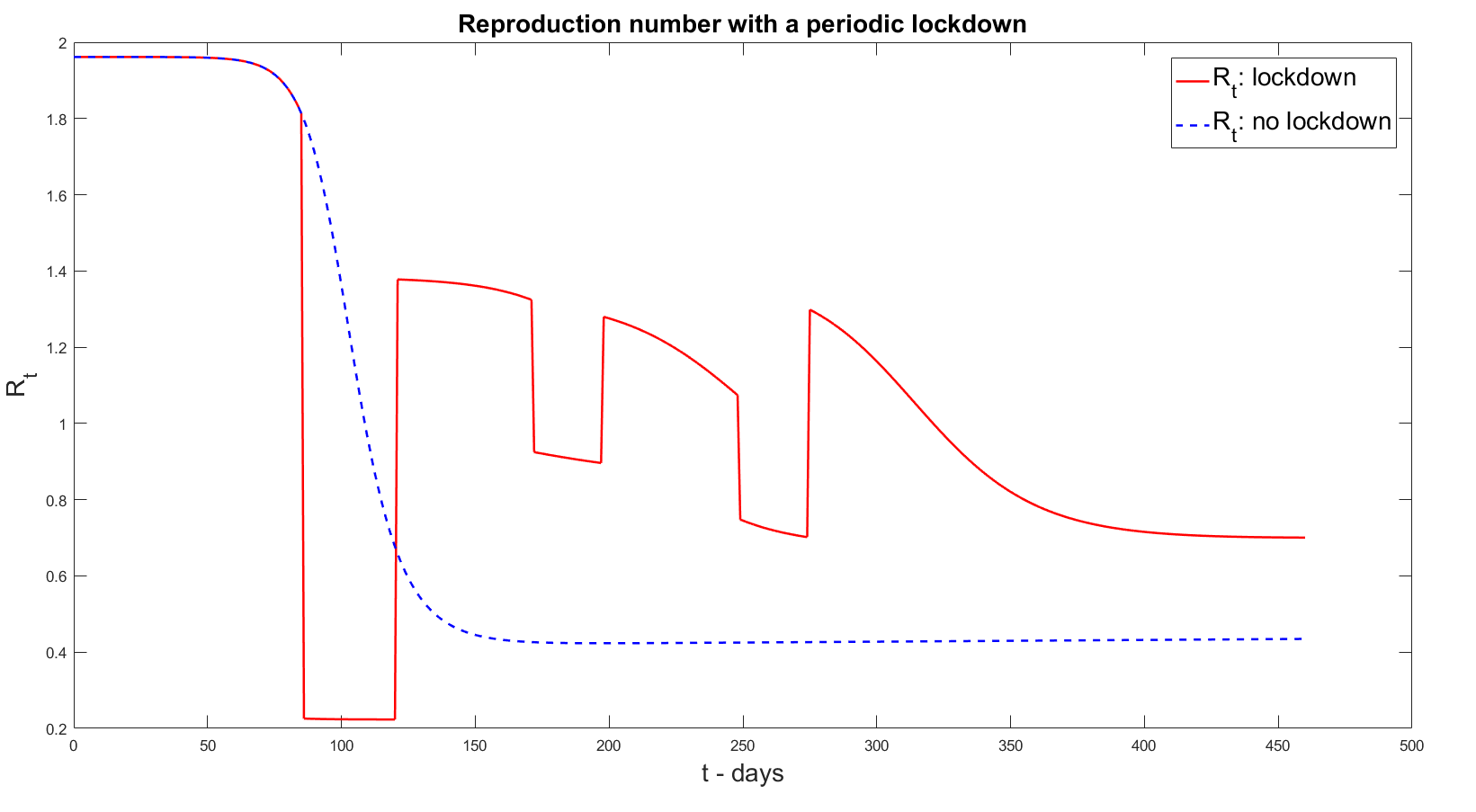}
    \caption{Reproduction number in the case of the optimization of three parameters.}
    \label{reproduction_exp4}
\end{figure}

\begin{table}[htbp]
    \centering
  \caption{Periodic containment: optimization
  over opening level at containments, duration of containment,
  level of reopening.}
%  \centering
%  \caption{Some numerical results of a periodic containment, when we optimize the loss functional with respect to three parameters. The numerical results provided here comes from the results of charts in Figure \ref{comparison2_exp4}.}
    \begin{tabular}{lcccc}
     & \multicolumn{1}{c}{Epidemy} & \multicolumn{1}{c}{ 
         Stricter policy} & \multicolumn{1}{c}{Optimal policy}& \multicolumn{1}{c}{Mild policy} \\
          & \multicolumn{1}{c}{No policy  } & \multicolumn{1}{c}{Fig. \ref{comparison2_exp4}.B  } & \multicolumn{1}{c}{Fig. \ref{comparison2_exp4}.A } & \multicolumn{1}{c}{Fig. \ref{comparison2_exp4}.C} \\ \hline \vspace{0.1cm}
     Successive opening levels $c$ & 100\%  & 80\%  & 81.4\%  & 85\%  \\  \vspace{0.1cm}
    Level of reopening & 100\%  & 85\%  & 91.8\%  & 93\%  \\  \vspace{0.1cm}
    Optimal period of closure (in days) & / & 25 & 25 & 25 \\ \vspace{0.1cm}
   Mortality at Day $85$ & 0.03\%   & 0.03\%   & 0.03\%   & 0.03\%  \\ \vspace{0.1cm}
    Total mortality at Day $460$ & 1.03\%   & 0.97\%   & 0.84\%   & 0.78\%  \\ \vspace{0.1cm}
    Mortality reduction & 0\%  & 5.76\%  & 17.85\%  & 24.27\%  \\ \vspace{0.1cm}
    Annualized 1st quarter GDP loss & 2.42\%   & 4.30\%   & 4.30\%   & 4.30\%  \\ \vspace{0.1cm} 
    Total annualized GDP loss & 1.78\%  & 10.63\%   & 8.91\%   & 8.18\%  \\ \vspace{0.1cm}   
     Value loss functional & 129.53  & 102.89  & 102.6  & 102.14 \\ \hline    
    \end{tabular}%
  \label{tab:exp3_case_3_par}%
\end{table}%

\vskip 5cm
\newpage

\section{Sensitivity Analysis}
We provide, in this section, a Sensitivity Analysis (SA) to evaluate how some of the parameters influence the minimum of the loss functional. Initially, SA is performed by a \textit{global sensitivity analysis} approach using the Sensitivity Analysis tool of Matlab. Then, we also provide a local sensitivity analysis where we calculate the optimal policy varying one parameter at a time. 

The global sensitivity approach uses a representative set of samples of parameters to evaluate the loss functional, which includes also the level of lockdown or reopening  depending on the numerical experiment under investigation (see previous sections \ref{SectUniqLock}, \ref{sec:opt_reopening} and \ref{sec:opt_periodic_cont}). The workflow is as follows:
\begin{enumerate}
    \item For each parameter, including the opening level
    $\bar{c}$ during
    containment or reopening, we generate multiple values that the parameters can assume, namely we define the parameter sample interval by specifying a uniform probability distribution for each parameter. We create 200 combinations of these parameters.  
    \item\label{item:sa} Then, find the solution of the SEAIRD model and evaluate the loss functional at each combination of parameter values and choose the combination which gives the minimum value of the loss functional.
    \item\label{item:sa_plus_opt} Fixing the ``best outcome'' combination found in (\ref{item:sa}), except the opening levels $\bar{c}$, we run again the optimization procedure, used in the previous sections, to find the optimal value of the opening levels for that combination of parameters.       
\end{enumerate}
 Table \ref{tab:interval} indicates the ranges
 for each parameter.	
\begin{table}[htbp]
	\centering
	\caption{Ranges utilized for the sensitivity analysis}
	\begin{tabular}{lc}
	     \multicolumn{1}{c}{Parameter } & \multicolumn{1}{c}{Range}  \\ \hline \vspace{0.1cm}
		 $\delta$ & $[0.0014, 0.0028]$  \\   \vspace{0.2cm}
		$s$ & $[0.05, 0.15]$   \\   \vspace{0.2cm}
		$r$ & $[0, 0.05]$   \\   \vspace{0.2cm}
		$\sigma$   & $[1.01, 4]$ \\   \vspace{0.2cm}
		$\theta$  & $[\frac{10}{35}, \frac{1}{2}]$  \\   \vspace{0.2cm}
		$a$  & $[8000, 20000]$  \\   \vspace{0.2cm}
		$\bar{c}$ & $[0.5,1]$ \\  \hline \vspace{0.2cm}
	\end{tabular}%
	\label{tab:interval}%
\end{table}%
As expected, the parameter that carries a greater weight on the functional value is represented by $r$. In fact, this is clear in Figure \ref{parameter_influence}, where, as a result of the sensitivity analysis, a tornado plot is displayed. The coefficients are plotted in order of influence of parameters on the loss functional, starting with those with greatest magnitude of influence from the top of the chart.  

Below, for each numerical experiment, we provide a table 
comparing results from the optimal case determined
with our methods, and the optimal case after the
SA described in  (\ref{item:sa}) and (\ref{item:sa_plus_opt})
above.  
For completeness, we also provide a local sensitivity analysis which is a technique to analyze the effect of one parameter on the cost function, and especially on the optimal policy. We take into account, as prototype, the first experiment where the optimal level of lockdown has to be found. See Table \ref{tab:local_sen_case1}. 

\begin{table}[htbp]
  \centering
  \caption{Comparison of the optimal values of Section \ref{SectUniqLock} with the result of 
  the global sensitivity analysis
 }
  
    \begin{tabular}{lccc}
          &  \multicolumn{1}{c}{Optimal case - Section \ref{SectUniqLock} } & \multicolumn{1}{c}{Optimal case - SA } \\ \hline \vspace{0.1cm}
    Containment and reopening level $\bar{c}$ & 76.7\% & 75.3\%  \\  \vspace{0.1cm}   
   Mortality at Day $85$ & 0.03\%  &  0.02\%    \\ \vspace{0.1cm}
    Total mortality at Day $460$ & 0.26\% & 0.21\%    \\ \vspace{0.1cm}
    Mortality reduction & 74.85\%  & 74.71\%  \\ \vspace{0.1cm}
    Annualized 1st quarter GDP loss & 3.29\%  & 3.30 \%   \\ \vspace{0.1cm} Total annualized GDP loss & 19.45\% & 20.55\%   \\ \vspace{0.1cm}   
     Value loss functional & 75.82 & 29.61  \\ \hline    
    \end{tabular}%
  \label{tab:sens_case1}%
\end{table}%

\begin{table}[htbp]
  \centering
  \caption{Comparison of the optimal situation of Section \ref{sec:opt_reopening} and of sensitivity analysis}
    \begin{tabular}{lccc}
          &  \multicolumn{1}{c}{Optimal case - Section \ref{sec:opt_reopening} } & \multicolumn{1}{c}{Optimal case - SA } \\ \hline \vspace{0.1cm}
    Reopening level $\bar{c}$ & 90.1\% & 93.9\%  \\  \vspace{0.1cm}   
   Mortality at Day $85$ & 0.03\%  & 0.03\%      \\ \vspace{0.1cm}
    Total mortality at Day $460$ & 0.63\% & 0.64\%    \\ \vspace{0.1cm}
    Mortality reduction & 38.43\% & 22.02\%  \\ \vspace{0.1cm}
    Annualized 1st quarter GDP loss & 4.30\%  & 4.31\%    \\ \vspace{0.1cm} Total annualized GDP loss & 12.02\%  & 9.41\%   \\ \vspace{0.1cm}   
     Value loss functional & 102.17 & 33.47  \\ \hline    
    \end{tabular}%
  \label{tab:sens_case2}%
\end{table}%

\begin{table}[htbp]
  \centering
  \caption{Comparison of the optimal situation of Section \ref{sec:opt_periodic_cont} and of sensitivity analysis}
    \begin{tabular}{lccc}
          &  \multicolumn{1}{c}{Optimal case - Section \ref{sec:opt_periodic_cont} } & \multicolumn{1}{c}{Optimal case - SA } \\ \hline \vspace{0.1cm}
    Successive reopening levels $\bar{c}$ & 72.9\% & 70.4\%   \\  \vspace{0.1cm}   
   Mortality at Day $85$ & 0.03\%  &  0.02\%    \\ \vspace{0.1cm}
    Total mortality at Day $460$ & 0.77\% & 0.63\%   \\ \vspace{0.1cm}
    Mortality reduction & 25.15\% &  24.98\%  \\ \vspace{0.1cm}
    Annualized 1st quarter GDP loss & 4.32\%  & 4.28\%    \\ \vspace{0.1cm} Total annualized GDP loss & 7.44\%  & 7.58\%   \\ \vspace{0.1cm}   
     Value loss functional & 107.41 & 39.56  \\ \hline    
    \end{tabular}%
  \label{tab:sens_case3}%
\end{table}%

\begin{figure}[h!]
    \centering
    \includegraphics[scale=0.3]{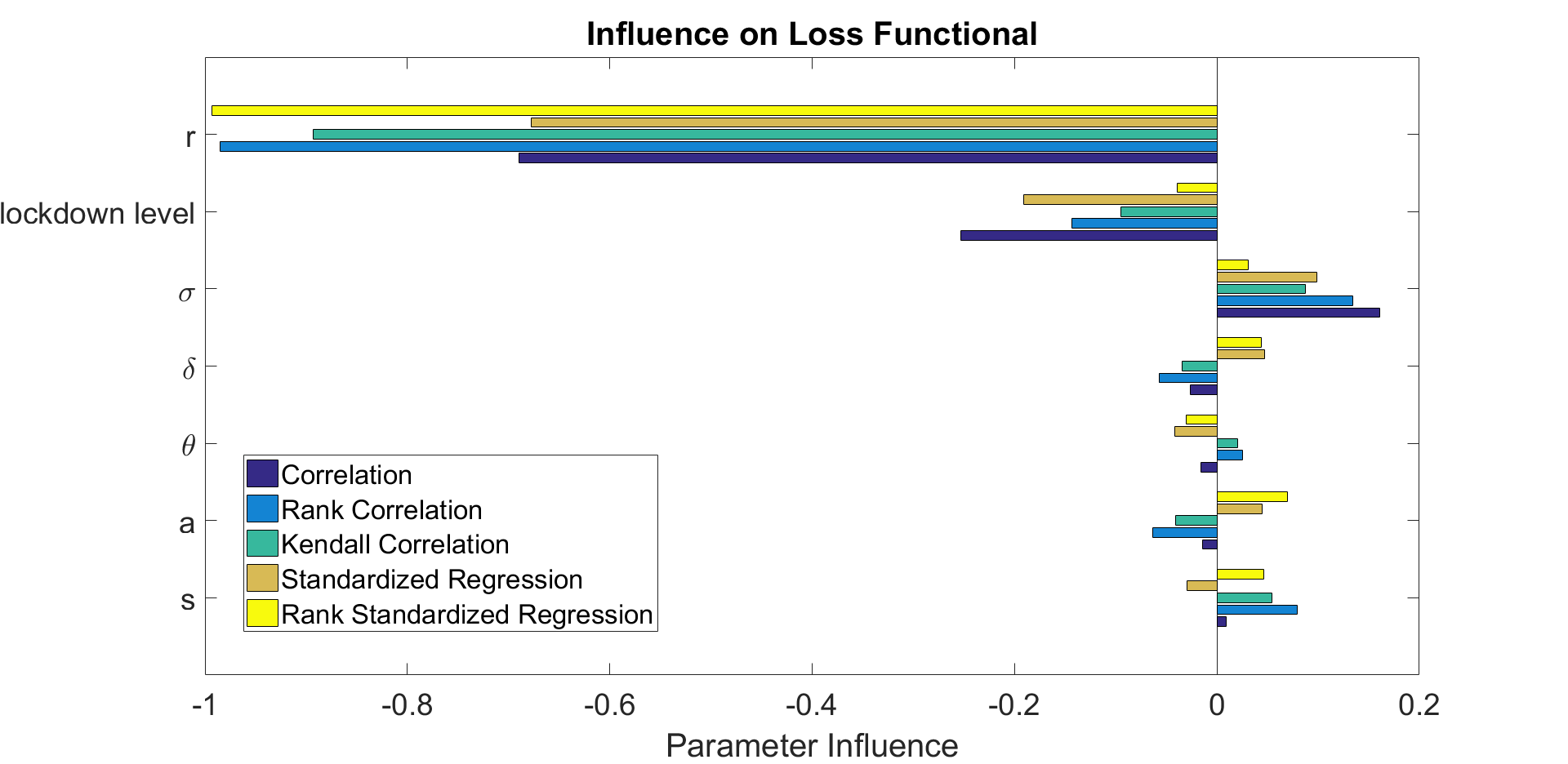}
    \caption{An example of the influence of the parameters on the loss functional for the first numerical experiment, with one unique lockdown.}
    \label{parameter_influence}
\end{figure}

\begin{table}[htbp]
	\centering
	\caption{Local sensitivity analysis for optimal unique lockdown experiment.}
	\begin{tabular}{lccc}
	     \multicolumn{1}{c}{Parameter } & \multicolumn{1}{c}{Range} & \multicolumn{2}{c}{   Optimal opening level at containment} \\         
	     &  & \multicolumn{1}{c}{at min range} & \multicolumn{1}{c}{at max range}  \\ \hline \vspace{0.1cm}
		 $\delta$ & $[0.0014, 0.0028]$ & 81.4\% & 76.7\%  \\   \vspace{0.2cm}
		$s$ & $[0.05, 0.15]$ & 78.7\% & 75\%    \\   \vspace{0.2cm}
		$r$ & $[0, 0.05]$ & 100\% & 77.2\%   \\   \vspace{0.2cm}
		$\sigma$   & $[1.01, 4]$ & 74.4\% & 80.3\% \\   \vspace{0.2cm}
		$\theta$  & $[\frac{10}{35}, \frac{1}{2}]$ & 79.7\% & 72.6\%   \\   \vspace{0.2cm}
		$a$  & $[8000, 20000]$ & 78.4\% & 71.4\%  \\  \hline \vspace{0.2cm}
	\end{tabular}%
	\label{tab:local_sen_case1}%
\end{table}%

\newpage
\section{Conclusions}

In this paper, we have formalized the trade-offs involved in the decision making between preserving economic activity and reducing the speed of diffusion of the pandemic. Our premise is that individual agents, as well as governments, want to contain and, possibly, postpone the infection and therefore the risk of a greater number of potential deaths to a later stage ("flatten the curve") in the expectation of better treatments, or a weakening of the virus, or a vaccine; we 
assume that actions are planned over a relatively short time horizon, that we choose to be 460 days.  
Our second working assumption is that there is a strong link between the degree of diffusion of the epidemic and the intensity of the economic shock\color{black}, with an elasticity that varies in time and across countries but seems to be in a range around 1/3. This elasticity is the result of all changes in behavior of agents, from the economic lockdown itself to the greater precautions of consumers who reduce their consumption and firms who favor drastic reductions in working time. 
 We have modeled containment measures by a
 function describing the level of opening, which we have taken to be piece-wise linear,
with additional regularities, to include feasibility;
we then formally described the trade-off between mortality reduction
and limitation of economic loss which
includes an estimation of the social cost $a$ of COVID-19 mortality,
and a discount rate which intensifies the 
effect of early deaths and early economic losses. 
\color{black} We discussed the
mathematical set-up and  proved the existence of at least
one optimal containment strategy.
A parametric representation of  mortality vs.
economic losses illustrates the potentialities
of the optimization approach.

Optimal control theory helps to find the right balance between the contrasting 
welfare needs during the COVID-19 epidemic,
even when limited to few free parameters. It also sheds light on the possibility of a non uniqueness of the
optimal control policy, very likely due to
the non convexity of the loss functional.
For instance, a  transition of phases takes place 
in terms of the parameter $a$
describing the social cost of COVID-19 mortality: at 
critical values of $a$, we  observed
the possibility of bifurcation towards two local optima and, therefore, discontinuous changes in the optimal policy as
function of $a$.
At and below the bifurcation, complete laissez-faire
is optimal, but it is never preferred when the statistical value of a life is large enough.
%We suggest that the initial conflicting proposals in several countries, such as Britain, US and Brazil, could be explained in terms of diverse perceptions of the fundamental parameters of the pandemic  prompting to envision interventions in the two regimes.\color{blue} \textit{[ Do we need this last sentence? ]}

\color{black}

\medskip

Given that, for most countries,  the
implied value of the social cost of COVID-19 death $a$
is in a range in which laissez-faire is not a viable solution, we discussed the optimal policies in a restricted set where the opening level can vary only a very limited number of times and where the solution  turns out to be unique. \color{black} 
Parameters have been estimated from available data, and
a sensitivity analysis has been carried out on the main ones.
 We have analyzed various examples: one unique lockdown
 to be extended till the presumed end of the 
epidemic at the end of the first quarter 2021,
a strategy that apparently very few countries tried to plan; 
 a drastic, initial lockdown, followed by 
 a reopening, which is what most countries
 are currently putting in place;
 some alternation of containment and reopening
 after the current one, 
 which is a plausible outcome if the
regained activity leads to recurrence of the virus.
The results  shed some light on the trade-offs involved, and suggests that gradual policies of longer duration but more moderate containment have large welfare benefits.
On the other hand, after a sharp lockdown has been put
in place, an alternation of containment and
reopening is worth of consideration.
\color{black}

\medskip
Finally, we have investigated the sensitivity of the
results on the estimated parameters. For most parameters,
our results are insensitive to moderate errors 
in their selection. Among the significant ones, the
most relevant has turned out
to be the discount rate: this reflects the belief that
early economic loss is more damaging, and that early deaths
 harm the health system and miss the opportunity of
 some form of adaptation to the virus or
 more effective treatments. In the
 examples we have considered a very high
 value for the discount rate, as we believe that 
treatment improvements are very likely.
It follows that  the timing is key to successful implementation of a containment policy, and that this is
 closely tied to the
 the pace and of the perspectives of potential technological advancement.
 \\
\newline
\newline
\textit{New York University in Abu Dhabi, May 28, 2020.}\footnote{ We wish to thank Christian Gollier and Benjamin Moll for comments on this draft.}
 \color{black}

\clearpage
\small{
\bibliographystyle{apalike}
\bibliography{biblio.bib}
}
\newpage

\section*{{Appendix 1: A problem with capital and consumption with fixed
saving rate (Solow type)}}

\subsection*{{Social planner's objective}}

{The social planner's utility $\mathcal{W}$ combines
now consumption defined below. The social planner minimizes a loss
function between an initial period $t=0$ and final period $T$ which could be infinity:}

{
\[
\mathcal{L}=\left\{ \intop_{0}^{T}e^{-rt}\left[\mathcal{V}(C(t))+aD'(t)\right]dt\right\} 
\]
}

\subsection*{{Economy }}

{Production combines labor, capital and the lockdown
control strategy:}

{
\[
P=F[c(.),L,K]
\]
 where $F$ is a Cobb-Douglas of each input with capital elasticity
$\alpha$. } {Note that here, the lockdown control only
affects labor utilization, one could also put it outside the
labor block but this is equivalent here.}

{Consumers save an exogenous fraction $\sigma$ of
output and use it to invest in capital. They also consume the rest,
that is, }

{
\[
C(t)=(1-\sigma)P(t)
\]
}

\textit{{NB: as in \cite{jones2020optimal}, it is possible to add a lockdown control $c_{c}(t)$ on the transformation of production
into consumption: one forces agents to stop consuming and this reduces
$\beta$.}}

{Capital stock is accumulated thanks to savings and
depreciates at rate $\mu$ say 10\% yearly and so follows:
\begin{align*}
\frac{dK}{dt} & =-\mu K+\sigma F[c(t)L,K]
\end{align*}
}

{There is still a link between GDP and transmission,
the lockdown policy is denoted by $c(t)$:
\[
\beta_{t}=\bar{\beta}c(t)
\]
}

\subsection*{Optimal control problem}

The epidemic part is kept identical but adds one control $C(t)$ and one constraint:
$$
\lambda_{K}\left[-\mu K+\sigma F[c(t)L,K]\right]
$$

\section*{Appendix 2: A Ramsey first best problem}

Now, let consumption be endogenous too, so that the saving rate is
not constant.

{The social planner's utility $\mathcal{L}$ combines
now consumption defined below. The social planner minimizes the same
loss function as before, between an initial period $0$ and a
final period $T$ which could be infinity:}

{
\[
\mathcal{L}_{C(t),c(t)}=\left\{ \intop_{0}^{T}e^{-rt}\left[\mathcal{V}(C(t))+aD'(t)\right]dt\right\} 
\]
but now has two instruments: one is the lockdown control c(t), the
second one is the consumption by agents $C(t)$, which determines at which
rate the capital can be accumulated, namely:}

\[
\frac{dK}{dt}=-\mu K+F[c(t)L,K]-C(t)
\]
\begin{rem} Existence of an optimal control can be shown similarly as for the case treated in the paper in both examples. In fact, in the first example the functional is clearly continuous in $c$ since consumption $C$ and capital $K$ are continuous in $c$ and consumption can be assumed, without loss of generality, to be bounded below by a positive constant $C_0$. In the second case we can prove existence of a minimizing pair $c^*,C^*$ by compactness. In fact, $c\in \mathcal{K}$ and, by the properties of $c$ and of the solutions of the SEIARD model, consumption $C$ is a uniformly Lipschitz continuous. Also, we can assume that consumption takes values in a closed and bounded interval $[C_0,C_1]$. This allows us to minimize the functional over a compact subset of $C[0,T]\times C[0,T]$ where $C[0,T]$ indicates the space of continuous functions on the interval $[0,T]$. Finally, using the continuity of the functional with respect to $c,C$ the existence of an optimal pair $c^*,C^*$ follows.    
\end{rem}

\begin{rem} In other problems, such as the Ramsey second best problem, the social planner may not be able to allocate consumption properly. Instead, private agents in a market economy choose themselves their consumption, maximizing their own utility function, leading to an arbitrage between consumption in different dates, corresponding to the traditional Euler equation in macroeconomics. This constraint is an additional constraint to the social planner and at this stage, our results do not apply to them. 
\end{rem}

\color{black}
\newpage

\section*{Data Appendix}

\vspace{0.4cm}
\begin{table}[htbp]
  \centering

\caption{Fatality rates by age.}
\begin{tabular}{ccccc}
\hline 
\multicolumn{3}{c}{} &  & \tabularnewline
\hline 
\hline 
\multicolumn{5}{c}{{\footnotesize{}Fatality rates per age (in \%)}}\tabularnewline
{\footnotesize{}Age groups} & {\footnotesize{}China   \cite{verity2020estimates}} &  & {\footnotesize{}France \cite{salje2020estimating}} & \tabularnewline
\hline 
{\footnotesize{}0-9} & {\footnotesize{}0.00161} &  & {\footnotesize{}0.001 } & \tabularnewline
{\footnotesize{}10-19} & {\footnotesize{}0.00695} &  & {\scriptsize{}(for 0-19)} & \tabularnewline
{\footnotesize{}20-29} & {\footnotesize{}0.0309} &  & {\footnotesize{}0.007} & \tabularnewline
{\footnotesize{}30-39} & {\footnotesize{}0.0844} &  & {\footnotesize{}0.02} & \tabularnewline
{\footnotesize{}40-49} & {\footnotesize{}0.161} &  & {\footnotesize{}0.05} & \tabularnewline
{\footnotesize{}50-59} & {\footnotesize{}0.595} &  & {\footnotesize{}0.2} & \tabularnewline
{\footnotesize{}60-69} & {\footnotesize{}1.93} &  & {\footnotesize{}0.8} & \tabularnewline
{\footnotesize{}70-79} & {\footnotesize{}4.28} &  & {\footnotesize{}2.2} & \tabularnewline
{\footnotesize{}80+} & {\footnotesize{}7.8} &  & {\footnotesize{}8.3} & \tabularnewline
 &  &  &  & \tabularnewline
{\footnotesize{}Less than 60} & {\footnotesize{}0.145} &  & {\footnotesize{}na} & \tabularnewline
{\footnotesize{}More than 60} & {\footnotesize{}3.28} &  & {\footnotesize{}na} & \tabularnewline
{\footnotesize{}Overall} & {\footnotesize{}0.657} &  & {\footnotesize{}0.53} & \tabularnewline
\hline 
\multicolumn{5}{c}{{\tiny{}Note: These figures refer to the ratio of probable deaths to infected population.}}\tabularnewline
\end{tabular}
  \label{tab:Fatality_Rates_Age}%
  
\end{table}

\vspace{0.4cm}

\end{document}